\documentclass[reqno,12pt]{amsart}

\usepackage{enumerate}
\usepackage[margin=1in]{geometry}
\usepackage{enumitem}
\usepackage{ifpdf}
\usepackage{amsmath}
\usepackage{amsfonts}
\usepackage{amssymb}
\usepackage{mathrsfs}
\usepackage[hyperfootnotes=false]{hyperref}
\usepackage{setspace}
\usepackage{amsrefs}
\usepackage{amsmath,amscd}
\usepackage{verbatim}
\usepackage[titletoc,title]{appendix}
\usepackage[dotinlabels]{titletoc}
\usepackage{amsthm}

\usepackage{setspace}
\usepackage{amsrefs}
\usepackage[all,cmtip]{xy}
\usepackage{nicefrac}

\newlength\Li \newlength\Lii 
\newcommand{\ignore}[1]{}

\newtheorem{thm}{Theorem}[section]

\newtheorem{cor}[thm]{Corollary}

\newtheorem{lemma}[thm]{Lemma}

\theoremstyle{definition}
\newtheorem{defn}[thm]{Definition}

\theoremstyle{remark}
\newtheorem{remark}[thm]{Remark}

\theoremstyle{note}

\newtheoremstyle{defnopunct}
{3pt}
{3pt}
{}
{}
{\bfseries}
{ }
{.5em}
{}

\theoremstyle{defnopunct}
\newtheorem*{hspnonum}{Hidden Subgroup Problem (HSP)}

\author{Asif Shakeel}
\address{Department of Physics, Haverford College, 
Haverford, PA 19041-1392, USA}
\email{ashakeel@haverford.edu}
\date{\today}
\ifpdf
\fi

\title[An improved query for the Hidden Subgroup Problem] {An improved query for the Hidden Subgroup Problem}

\begin{document}

%\doublespace

\begin{abstract}
An equal superposition query with $\vert 0 \rangle$ in the response register is used in the ``standard method"   of single-query algorithms for the  hidden subgroup problem (HSP). Here we introduce a different query, the character query, generalizing  the well-known phase kickback trick. This query maximizes the success probability of subgroup identification under a uniform prior,  for the HSP in which  the oracle functions take values in a finite abelian group. We then apply our results to the case when the subgroups are drawn from a set of conjugate subgroups and obtain a  success probability greater than that found by Moore and Russell.  
\end{abstract}

\maketitle

%\enlargethispage{\baselineskip}

%%%%%%%%%%%%%%%%%%%%%%%%%%%%%%%%%%%%%%%%%%%%%%%%%%%%%%%%%%%%%%%%%%%%%%%%%%%%%%%%%%%%%%%%%%%%%%%%%%%%%%%%%%%%

\section{Introduction} \label{section:intro}
%%%%%%%%%%%%%%%%%%%%%%%%%%%%%%%%%%%%%%%%%%%%%
The Hidden Subgroup Problem (HSP) has been the focal point of investigations into quantum algorithms for some time. Motivation for this direction of research originates in  the  wide variety of problems in classical computation that can be formulated as or can be reduced to instances of the HSP~\cite{nc:qcqi}. Some had been solved or were at least problems of great interest, like the  prime factorization of an integer, prior to the HSP formulation. In that and other such cases involving the abelian group $\mathbb{Z}/(N)$, the Quantum Fourier Transform of Shor~\cite{sp:ptfd} proved an invaluable tool for solving the HSP efficiently. Not all  problems involving the HSP  have underlying groups that are abelian. There are  problems that have significance in classical computation for independent reasons, with  connections to non-abelian groups.  Among these,  Ettinger and H\o yer~\cite{eh:qogi} have considered  the graph isomorphism problem,  and Regev~\cite{ro:qclp} has considered the search for the shortest vector in a lattice (SVP),  reducing these problems   to  HSPs over the symmetric group and the dihedral group respectively. Thus, many families  of  groups have been studied as a result of the HSP assuming a central role as a problem of interest.

Over time, the statement of the problem has  evolved in a subtle way due, in part, to the availability of ideas in related areas of  decision and estimation theory.  In the earlier version of the problem, an unknown subgroup is hidden by an oracle function,  provided  via responses to queries. The goal is to determine the subgroup  with maximum probability using a polynomial number of queries,  with  polynomial time computations, and utilization of polynomial space resources in the size of the group (i.e., efficiently). More recently, a different version is also being pursued, one that arises when average case solutions are more relevant.  In this version, the aim  is to find with highest probability the correct  subgroup from  an ensemble of subgroups that the oracle could hide. This ensemble is assigned  a prior probability distribution  giving  the likelihood of a particular subgroup being the  hidden subgroup.  A measurement and a decision algorithm are chosen to maximize the detection  probability. This moves the problem to the domain of estimation, where a state discrimination approach originally proposed by Holevo~\cite{ah:sdqs}, and Yuen, Kennedy, and Lax~\cite{ykl:opt} comes to the fore. This is more of a measurement based approach, in which a measurement is optimized to glean the most information for the decision stage. This approach to the HSP originated in the work of Ip~\cite{ip:sao}. His result showed the optimality of Shor's algorithm, for an abelian group and under the assumption of equal prior on all the possible subgroups.

There is a wide variation in  approaches to solve the HSP.   Kitaev~\cite{ak:qmas} and Mosca and Eckert~\cite{me:hspee} have done seminal work concerning  abelian groups,  and similarly  Ettinger and H\o yer~\cite{eh:qanc}, Kuperberg~\cite{kup:sed} and  Regev~\cite{reg:sep} have done work of great significance on the dihedral group.    As the choice of measurement strongly affects the performance of an algorithm,  Bacon, Childs and Van Dam~\cite{bcv:opd,bcv:omqa}, Moore and Russell~\cite{cmar:pgm}, Bacon and Decker~\cite{bd:opscm}, and Krovi and R{\"o}teller~\cite{kr:eqwh}, have derived important results on optimal measurements, both for general finite groups and for specific families of non-abelian groups.  In the context of state discrimination, several of these  algorithms have employed a particular measurement with much success: the Pretty Good Measurement (PGM)~\cite{cmar:pgm,bcv:opd,bcv:omqa}. The PGM will be used in our treatment as well.  It is often best viewed  in the Fourier domain. Besides the PGM implementation, one encounters the non-abelian Fourier transform  quite commonly in HSP algorithms, which are surveyed in~\cite{fw:hrp,cl:hrp}. We see applications of Fourier transform in HSP problems such as that in Hallgren, Russell, and Ta-Shma~\cite{hrt:nsr},where they  use it for  normal subgroup reconstruction, and in a host of other HSP algortihms~\cite{cmar:pgm,bd:opscm,kr:eqwh,db:cghg,mrrs:psfs}. This has  naturally led to efforts to find efficient algorithms to implement Fourier transform,  which is itself a non-trivial problem.  Among others,  Hales and Hallgren~\cite{hh:qft} have results for abelian groups, and Moore, Rockmore and Russell ~\cite{mrr:psfs} have those for non-abelian  groups.

We are concerned with the state discrimination version of the HSP, continuing to approach it as an optimization problem. Instead of keeping the query independent of the structure of the problem, however, we take the first step in the direction of optimizing the query, making it into a joint query-measurement optimization. Consistent with most of the literature, we impose an abelian structure on  the set in which the oracle function takes values. This allows us to get more control over the query, though not completely since there is no  assumption concerning the group to which the hidden subgroup belongs, except that it is a finite group.  We discover that with that simple constraint, the query that works  best  is one that generalizes the \textit{phase kickback trick}. The latter was first used to solve Deutsch's problem, namely, to find the parity of a function on the set $\{0,1\}$. As we will explain in the next section, Deutsch's problem is actually an HSP. 

We begin our discussion with a  brief introduction to  Deutsch's problem, which illustrates in key ways the ideas of this paper. These ideas are summarized, together with our assumptions and results, in  subsection~\ref{subsection:qs}, under the title \textbf{query search}.

 %%%%%%%%%%%%%%%%%%%%%%%%%%%%%%%%%%%%%%%%%%%%%%%%%%%%%%%%%%%%%%%%%%%%%%%%%%%%%%%%%%%%

\subsection{The phase kickback trick} \label{subsection:pk}

The Deutsch-Jozsa algorithm~\cite{dj:rspq,rj:qaft}  is one of the earliest examples of  quantum algorithms.  It  determines if a function:
\begin{equation*}
f: {\mathbb{Z}/(2)}^{\times n} \rightarrow \mathbb{Z}/(2)
\end{equation*} is constant or ``balanced'' ($0$ on exactly half of the domain). The phase kickback trick~\cite{cemm:qar} version of the algorithm solves  the problem with one oracle query. Its underlying problem (Deutsch's problem) for single bit ($n=1$)  is an instance of the hidden subgroup problem (HSP). Being a basic example, it provides a  good motivation for the rest of the discussion and a chance to familiarize the reader with some concepts and definitions that will come up.

We follow the path in~\cite{rj:qaft} to recall the algorithm and  compare  the ``standard method" (as defined in~\cite{mrrs:psfs}) with  the ``phase kickback trick".  

We are given  an oracle $\textbf{O}_f$ that implements the function $f$ as a unitary transformation on $\mathbb{C}[\mathbb{Z}/(2)] \otimes \mathbb{C}[\mathbb{Z}/(2)]$:
\begin{equation*}
\textbf{O}_f:\vert x\rangle\vert y\rangle \mapsto\vert x\rangle\vert y + f(x)\rangle 
\end{equation*}
The first qubit on which the oracle evaluates the function $f$ is the ``query register", and the second qubit to which the oracle evaluation is added is the ``response register".

In both  the methods,  the initial state, or the ``query",  is of the form: $\vert \Psi \rangle= \frac{1}{\sqrt{2}}(\vert 0 \rangle + \vert 1 \rangle)\vert v \rangle$.  In the standard method for Deutsch's problem  $\:\vert v \rangle $  is set to $\vert 0 \rangle$. In the phase kickback trick  $\:\vert v \rangle $ is  set to $ \frac{1}{\sqrt{2}}(\vert 0 \rangle - \vert 1 \rangle)$. It is this  response register part, $\:\vert v \rangle$,  of the query  that we will refer to as the ``slate" of the query. Let us write the two queries.
\begin{equation*}
\begin{array} {rll}
\vert \Psi_s \rangle &= \frac{1}{\sqrt{2}}(\vert 0 \rangle + \vert 1 \rangle)\vert 0 \rangle   \quad & \text{(the standard query)} \\
\vert \Psi_p \rangle &= \frac{1}{{2}}(\vert 0 \rangle + \vert 1 \rangle) (\vert 0 \rangle - \vert 1 \rangle)   \quad &\text{(the phase kickback query)} 
 \end{array}
\end{equation*}
 Recall the Hadamard transform $\textbf{H}$, Fourier transform on $\mathbb{Z}/(2)$. Its action on the computational basis is:
  \begin{equation*}
\textbf{H}:  \left \{ \begin{array} {ll}
\vert 0 \rangle &\mapsto \frac{1}{\sqrt{2}} (\vert 0 \rangle  +\vert 1 \rangle) =:\vert + \rangle \\
 \vert 1 \rangle &\mapsto \frac{1}{\sqrt{2}} (\vert 0 \rangle  -\vert 1 \rangle) =:\vert - \rangle 
  \end{array}
  \right .
\end{equation*} 
Rewrite the queries in terms of $\vert + \rangle$ and  $\vert - \rangle$:
\begin{equation*}
\begin{array} {rl}
\vert \Psi_s \rangle &=\vert + \rangle\vert 0 \rangle    \\
\vert \Psi_p \rangle &=\vert + \rangle\vert - \rangle    
 \end{array}
\end{equation*}

In the standard method,  the state  after the oracle evaluation is:
\begin{equation} \label{sms}
\textbf{O}_f \vert \Psi_s \rangle = \frac{1}{\sqrt{2}}\vert + \rangle \vert + \rangle   \:\:+ \:\:\frac{1}{2}\big( (-1)^{f(0)}\vert 0 \rangle +  (-1)^{f(1)}\vert 1 \rangle\big) \:\:\vert - \rangle  
 \end{equation}
whereas for the phase kickback trick, the state after oracle evaluation is:
\begin{equation} \label{pks}
\textbf{O}_f\vert  \Psi_p \rangle =  \frac{1}{\sqrt{2}} \big( (-1)^{f(0)}\vert 0 \rangle +  (-1)^{f(1)}\vert 1 \rangle \big)  \:\:\vert - \rangle
 \end{equation}

Next is the measurement. Consider the  state $\textbf{O}_f \vert \Psi_s \rangle $ arising from oracle evaluation on the standard query first.  An application of the Hadamard transform $\textbf{H}$ to each register (denoted by  $\textbf{H}\otimes \textbf{H}$), rotates the state  to:
 \begin{align*}
\textbf{H} \otimes \textbf{H} \left(\textbf{O}_f\vert \Psi_s \rangle \right)= \frac{1}{\sqrt{2}}\vert 0 \rangle\vert 0 \rangle   + \frac{1}{2\sqrt{2}} \Big[& \left((-1)^{f(0)} +  (-1)^{f(1)}\right)\vert 0 \rangle   \\
  &+ \left((-1)^{f(0)} -  (-1)^{f(1)}\right)\vert 1 \rangle\Big]\vert 1 \rangle 
 \end{align*}
Measure the response (second) register.  If the result is $0$, which occurs with probability $1/2$, there is no information about the function. Output either ``constant" or ``balanced"  (random guess), which would be correct with probability $1/2$.  If the result is $1$, then measure the  query (first) register.  If the  result  is $0$, output ``constant". If the result is $1$,  output  ``balanced". With overall probability $3/4$  we get the correct answer. 

Note that we can  work solely with the query register to get the same result through a slightly different route. Apply $\textbf{H}$ to the query register (denoted by $\textbf{H} \otimes \textbf{I}$,  where $ \textbf{I}$ is the identity operator on $\mathbb{C}[\mathbb{Z}/(2)]$). 
 \begin{align*}
\textbf{H} \otimes \textbf{I} \left(\textbf{O}_f\vert \Psi_s \rangle \right)= \frac{1}{\sqrt{2}}\vert 0 \rangle\vert + \rangle   + \frac{1}{2\sqrt{2}} \Big[& \left((-1)^{f(0)} +  (-1)^{f(1)}\right)\vert 0 \rangle   \\
  &+ \left((-1)^{f(0)} -  (-1)^{f(1)}\right)\vert 1 \rangle\Big]\vert - \rangle 
 \end{align*}
  Measure the query  register and output ``constant" if the result is $0$, otherwise output ``balanced". The algorithm succeeds with overall probability $3/4$ as before. So instead of making measurements on both the registers, we may as well just measure the query register. 

Now consider the  state after oracle evaluation on the phase kickback query , $\textbf{O}_f\vert  \Psi_p \rangle$.  Apply $\textbf{H}$ to the query register.
\begin{align*}
\textbf{H} \otimes \textbf{I}    \left(\textbf{O}_f\vert \Psi_p \rangle \right) &=  \frac{1}{2}\Big[\big((-1)^{f(0)} +  (-1)^{f(1)}\big)\vert 0 \rangle  +\big((-1)^{f(0)} -  (-1)^{f(1)}\big)\vert 1 \rangle\Big]\:\:\vert - \rangle 
 \end{align*}
Measure the query  register.  If the result  is $0$, output ``constant", otherwise output ``balanced". With probability $1$ we get the correct answer. The measurement is the same for both the methods and is confined to the query register. This measurement (including the Hadamard transform) of the query register is what we refer to as the ``measurement for Deutsch's problem". It can also be described as the measurement in the basis $\{\vert + \rangle,\vert - \rangle\}$ (without the Hadamard transform) of the query register. If the result is $\vert + \rangle$, output ``constant". If the result is $\vert - \rangle$,  output ``balanced".

This calculation is used to  compare the probability of success of the two  queries in  identifying the function as being constant or balanced. While demonstrating that,   it raises the question of how  in general  would one be able to derive such a query.   We approach the problem from the perspective of query search.  

%%%%%%%%%%%%%%%%%%%%%%%%%%%%%%%%%%%%%%%%%%%%%
\subsection{The hidden subgroup problem and the state discrimination approach} \label{subsection:hspsd}
Let us recall the general hidden subgroup problem as stated in~\cite{me:hspee}.
\begin{hspnonum} 
Let $G$ be a group,  $X$ a finite set, and  $f: G\rightarrow X$ a function. There exists a subgroup $H \leq G$ such that $f$ is constant and distinct on the cosets (assume left cosets) of $H$.  That is, the function  has the property:
\begin{equation} \label{ofunc}
f(gh)=f(g)  \quad \forall h \in H, \:\: g \in G
\end{equation}
and $f(g) = f(g') \iff g' \in gH$.  $f$ is accessed via queries to  an oracle. Using information gained from evaluations of $f$ via its oracle, determine a generating set for $H$.
\end{hspnonum}

We say that the function $f$  hides the subgroup $H$, and the oracle implements the function $f$. We call X  the ``response space", $\mathbb{C}[G]$  the ``query register" and $\mathbb{C}[X]$ the ``response register". Together they form the system $\mathbb{C}[G] \otimes \mathbb{C}[X]$ on which the oracle and measurement  act.

Within the standard method, optimal measurements have been determined in many cases using results from the state discrimination approach.   In this paradigm, a measurement is deemed optimal if it maximizes the probability of subgroup identification over some set of  subgroups of $G$  distributed according to a (usually uniform) prior probability.  Optimality can be confirmed by verifying simple checkable criteria due to Holevo~\cite{ah:sdqs}, and Yuen, Kennedy, and Lax~\cite{ykl:qar}. This technique has been shown to be very useful~\cite{cmar:pgm,bd:opscm,bcv:opd} in works  investigating  optimal measurements.  

In state discrimination, the objects to be distinguished are  a set of states $\{\rho_k\}_{k\in\mathscr{K}}$ on some finite dimensional Hilbert space $V$, indexed by a finite set $\mathscr{K}$, and distributed  according to  a probability function $\{p_k\}_{k \in \mathscr{K}}$ such that $p_k$ is the prior probability of occurrence of $\rho_k$. A measurement on $V$ is described by a set of operators (POVM) $\mathscr{E} := \{E_k\}_{k \in \mathscr{K}}$. These satisfy:  $E_k \geq 0$ and $\sum_{k \in \mathscr{K}} E_k = \textbf{I}$, where $\textbf{I}$ is  the identity operator. The measurement operator $E_k$ corresponds to the outcome $k$, associated to the state $\rho_k$.   The probability of successful state discrimination using a measurement  $\mathscr{E} =  \{E_k\}_{k \in \mathscr{K}}$ is:
\begin{equation} \label{sucs}
S(\mathscr{E}) :=    \sum_{k \in \mathscr{K}} p_k \:\: \text{tr} \: (E_k \rho_k)
\end{equation}	
The idea is to find a measurement that maximizes $S(\mathscr{E})$  over the set of all measurements. Such a measurement is deemed optimal.

In  Deutsch's problem for single bit, the relevant HSP is as follows. The group $G = \mathbb{Z}/(2)$.  The response space is $X = \mathbb{Z}/(2)$. Possible hidden subgroups are: $H_0 := G $ and $H_1 := \{0\}$. Constant functions hide the subgroup $H_0$ and  balanced functions hide $H_1$. The system is:
 \begin{equation*}
 \mathscr{H} :=  \mathbb{C}[G]\otimes\mathbb{C}[X] = \mathbb{C}[\mathbb{Z}/(2)]\otimes\mathbb{C}[\mathbb{Z}/(2)]
 \end{equation*}
State of the system after querying the oracle implementing a function $f$ in the standard method is given by \eqref{sms}. Since we are interested in measurements on the query register $\mathbb{C}[G]$,  we need only consider the reduced density operator  for the query register, as shown in Nielsen and Chuang~\cite{nc:qcqi}. This is found by taking the partial trace over the response register $\mathbb{C}[X]$: $\text{tr}_{_{\mathbb{C}[X]}}\left( \textbf{O}^{\phantom\dag}_f\vert \Psi_s \rangle \langle \Psi_s \vert \textbf{O}^\dag_f \right)$. If $f$ is constant, i.e., hides $H_0$, the mixed state (reduced density operator on the query register) is: 
\begin{equation*}
{\rho_s}_{_0} :=  \vert + \rangle \langle + \vert 
\end{equation*}
If $f$ is balanced, i.e., hides $H_1$, the mixed state is:
\begin{equation*}
{\rho_s}_{_1} :=   \frac{1}{2} \big(\vert + \rangle  \langle + \vert   \:\:+\:\: \vert - \rangle \langle - \vert\big) 
\end{equation*}

Bacon and Decker~\cite{bd:opscm} derive the optimal measurement for the  standard method single-query HSP for a finite group $G$,  assuming that there is a uniform prior probability on the set of all  subgroups of $G$ of being  hidden by the oracle function.  To underscore the issue, they show at the outset that  the aptly named Pretty Good Measurement (PGM)~\cite{cmar:pgm}, which is optimal in this sense for several important cases,  is sub-optimal for Deutsch's problem. 
The PGM, denoted by $\mathscr{M} = \{M_k\}_{k\in\mathscr{K}}$,  for the general states setting is defined as: 
\begin{align} \label{pgmdef}
M_k &:= p_k \:\: \rho^{-1/2} \rho_k \rho^{-1/2}, \quad \rho := \sum_{k \in \mathscr{K}} p_k \:\: \rho_k \nonumber \\
{\rho}^{-1/2} &:= \left(\rho\mid_{\text{Im}(\rho)}\right)^{-1/2}\oplus \textbf{I}_{\text{Ker}(\rho)}
\end{align}

For Deutsch's problem using the standard method with the assumption of uniform prior probability on the subgroups, the PGM, denoted by $\mathscr{M}_s = \{{M_s}_{_0}, {M_s}_{_1}\}$, becomes:
\begin{equation*}
{M_s}_{_0} =   \frac{2}{3}\vert + \rangle \langle + \vert  \quad \text{and} \quad  {M_s}_{_1} =   \frac{1}{3} \vert + \rangle  \langle + \vert   \:\:+\:\: \vert - \rangle\langle - \vert
\end{equation*}
which has the success probability $S_{d} = 2/3$ by \eqref{sucs}, less than the $3/4$ obtained by the standard method above.  The measurement for Deutsch's problem (defined in section \ref{subsection:pk} as the measurement on the query register for both the standard method and the phase kickback trick) written as a POVM, is:
\begin{equation*}
E_0 := \vert + \rangle \langle + \vert  \quad \text{and} \quad  E_1 :=  \vert - \rangle  \langle - \vert 
\end{equation*}
Let us take a look at  the phase kickback trick for Deutsch's problem, and compute the mixed states for this case using \eqref{pks}. If $f$ is constant, i.e., hides $H_0$, the mixed state is: 
\begin{equation*}
{\rho_p}_{_0}:=  \vert + \rangle \langle + \vert 
\end{equation*}
If $f$ is balanced, i.e., hides $H_1$, the mixed state is:
\begin{equation*}
{\rho_p}_{_1} :=   \vert - \rangle\langle - \vert 
\end{equation*}
It is apparent that the PGM, denoted by  $\mathscr{M}_p = \{{M_p}_{_0}, {M_p}_{_1}\}$, for this set of mixed states is given by:
\begin{equation*}
{M_p}_{_0} =    \vert + \rangle \langle + \vert = E_0 \quad \text{and} \quad  {M_p}_{_1} =   \vert - \rangle  \langle - \vert = E_1 
\end{equation*}
which is the measurement for Deutsch's problem, yielding a success probability of $1$, certainly optimal.  This contrasts with the standard method for which the PGM was found to be sub-optimal.  By definition, the PGM  depends on the set  of mixed states  to be distinguished, which in turn  depends both on the  possible subgroup  that the oracle hides and the query presented to the oracle.  The question of optimality of the PGM, or for that matter any measurement,  is  perhaps better posed in that context. It is interesting  that the  particular measurement $\mathscr{M}_p$ which is optimal  for the phase kickback query also works well for the standard query and is in fact optimal. We can verify this by the construction of Bacon and Decker~\cite{bd:opscm}, or by the checkable optimality conditions of Yuen, Kennedy, and Lax~\cite{ykl:qar}. Under some assumptions, we find that a wider set of queries have shared optimal measurements.

%%%%%%%%%%%%%%%%%%%%%%%%%%%%%%%%%%%%%%%%%%%%%
\subsection{Query search} \label{subsection:qs}
We build upon the paradigm of state discrimination hitherto adopted.  Much of the literature on the subject considers HSP for  a finite group $G$. This is also what we assume. We impose no restrictions on the set of possible hidden subgroups, letting that set be some arbitrarily chosen set of subgroups of $G$.  We work with  the following reasonable assumptions:  there is a uniform prior  probability on the  given set of subgroups of being hidden by the oracle function, and similarly there is a uniform prior probability that the oracle implements a  function from the set of  functions hiding a particular subgroup. The former is granted in~\cite{cmar:pgm,bd:opscm,bcv:opd}, and the latter is within the spirit of the HSP where the value of a function gives no information about the hidden subgroup. Although in the definition of the HSP, the response space $X$ can be any finite set, in the literature and the instances with which we are familiar,  it is  some abelian group of finite cardinality. For example, Regev~\cite{ro:qclp} takes advantage of this structure in his reduction of the shortest vector in a lattice problem (SVP) to HSP over the dihedral group. Having this structure on $X$ allows us to look for ways in which it can be exploited  to enhance the success probability of an HSP algorithm.  It is consistent with the oracle action definition for the  general  HSP. We consider the queries that are in an equal superposition state over the group and  arbitrary in the slate (state of the response register $\mathbb{C}[X]$), referring to these as the equal superposition tensor product or ESTP queries.  We consider algorithms in which the  measurements  are restricted to  the query  register $\mathbb{C}[G]$.   

Previous works~\cite{cmar:pgm,bd:opscm,bcv:opd} investigating the optimality of measurements for the standard method have used the maximization of the probability of subgroup identification as a criterion. We extend the criterion of optimality to ESTP queries: to be optimal, a query  need  maximize the probability of subgroup identification over all measurements and over all ESTP queries. It turns out that a generalization of the query used for phase kickback has the highest success probability. We call such a   query a ``character query": its slate  is  a particular character of $X$.

In section \ref{section:genhsp} we  describe the context, explicitly state the assumptions and  the class of algorithms to be used, and review some background material. In section \ref{section:clq}   we motivate and develop the main result concerning the mixed state obtained after the oracle evaluation on a query, assuming the oracle hides a particular subgroup.  Perhaps somewhat curiously, this state does not depend on which abelian group $X$ is, and depends only on its cardinality, i.e., the dimension of the response register.  In section \ref{section:cq} we give definitions of the ``standard query" and the ``character query". We find that any measurement can be taken to be an element of the group   algebra $\mathbb{C}[G]$ acting by the right regular representation.  Further, the success probability of any query is linearly dependent on that of the character query. More importantly, besides the phase multiples of a constant query, all ESTP queries have identical optimal measurements. This allows us to  reuse  previously known optimal measurements, specially those from the vast literature on  the standard query based HSP.  In section \ref{section:oq} we prove that the character query has the maximum probability of success, and does strictly better than the standard query. This gives us an explicit example of a query that satisfies our optimality criterion. We then take another look  at  the Deutsch's problem for single bit.  In the process we re-explain why phase kickback works as well as it does for Deutsch's problem in the single bit case from a representation theory point of view, with a clue to other problems in the HSP category. In section \ref{section:consub} we derive the success probability when the family of subgroups consists of conjugates. We find an improvement over the success probability found by Moore and Russell~\cite{cmar:pgm}, and describe how this case  contrasts with  Deutsch's problem. Section \ref{section:conclusion} is the conclusion and discussion section.

%%%%%%%%%%%%%%%%%%%%%%%%%%%%%%%%%%%%%%%%%%%%%%%%%%%%%%%%%%%%%%%%%%%%%%%%%%%%%%%%%%%%%%%%%%%%%%%%%%%%%%%%%%%%

\section{Background} \label{section:genhsp}

To begin, we define some terms that are relevant to this discussion. We have the following data: a  finite group $G$, a set  $\mathscr{S}=\{H_k\}_{k \in \mathscr{K}}$ of subgroups  of $G$ indexed by a finite set $\mathscr{K}$ of cardinality $ K := |\mathscr{K} |$. For a subgroup $H_k \in \mathscr{S}$, we denote its index  by $N_k := [G:H_k]$. The probability that the oracle function hides any particular subgroup $H_k \in \mathscr{S}$ is $1/K$.

We fix the response space $X$ to be some finite abelian group. Up to isomorphism, we can assume that $X = \mathbb{Z}/({p^{\alpha_1}_1})\times \ldots\times \mathbb{Z}/({p^{\alpha_m}_m})$ for some primes $p_i$ and some $m, \alpha_i \in \mathbb{Z}^{+}$. Let  $D := |X|$.  Then $D=p^{\alpha_1}_1\ldots p^{\alpha_m}_m$.  By necessity $D  \geq \max_{k \in \mathscr{K}}\{N_k\}$. Associated with each subgroup $H_k \in \mathscr{S}$, we have  the set $\mathscr{F}_k $ of all the oracle functions that   satisfy condition \eqref{ofunc} for the subgroup $H_k$.
\begin{equation*}
\mathscr{F}_k := \{f: G\rightarrow X \mid f \text{ is constant and distinct on the left cosets of } H_k\}
\end{equation*}
The probability that the oracle implements a function $f \in \mathscr{F}_k$ given that the function hides the subgroup $H_k$ is $1/|\mathscr{F}_k|$.  

We have a Hilbert space, $\mathscr{H}$, describing the composite system. It is the tensor product of the query register:  $\mathbb{C}[G]$, and the response register: $\mathbb{C}[X]$. 
\begin{equation*}
\mathscr{H} := \mathbb{C}[G]\otimes\mathbb{C}[X]
\end{equation*}

We introduce the computational basis of   $\mathbb{C}[G]$ consisting of the $\delta$ functions:  
\begin{equation*}
\delta_z(g) = \left\{
\begin{array}{lr} 
1, & g=z \\
0, & \text{otherwise}
\end{array}
\right .
\end{equation*}
We also write $\vert z \rangle := \delta_z$. Then an element $\vert \phi \rangle \in \mathbb{C}[G]$ has a unique expression as a sum $\vert \phi \rangle = \sum_{g\in G} \phi_g \vert g \rangle$. 

$\mathbb{C}[G]$ has an inner product:
\begin{equation*}
\langle \phi \vert \vartheta \rangle = \sum_{g \in G} \overline{\phi_g}\vartheta_g
\end{equation*}
where $\vartheta  = \sum_{g\in G} \vartheta_g \vert g \rangle$.

We  remind ourselves of translations of  $\mathbb{C}[G]$ in the manner of Goodman and Wallach~\cite{wall:sri}. We denote by   $L$  and $R$ the left and right translation representations of $G$ on $\mathbb{C}[G]$. On the element:
\begin{equation*}
\vert \phi \rangle = \sum_{g \in G} \phi_g\vert g \rangle \in \mathbb{C}[G]\text{,}
\end{equation*}
$z \in G$ acts by:
\begin{align} \label{rtact}
L(z)\vert \phi \rangle  &:= \sum_{g \in G} \phi_g\vert z g \rangle \\
R(z)\vert \phi \rangle  &:= \sum_{g \in G} \phi_g\vert gz^{-1} \rangle
\end{align}
This induces corresponding actions by the computational basis $\vert z \rangle \in \mathbb{C}[G]$. Extend   by linearity to  get the left and right  \textit{regular representations} of $\mathbb{C}[G]$ on itself, also denoted by $L$ and $R$:
\begin{equation*}  %\label{ltrtact}
L,R : \mathbb{C}[G] \rightarrow \text{End}(\mathbb{C}[G])
\end{equation*}
We have a similar description for    $\mathbb{C}[X]$ except that $X$ is abelian and we use ``$+$" to designate the group operation in $X$. We have an inner product on $\mathscr{H}$ (defined as for $\mathbb{C}[G]$) compatible with the tensor product structure. In the ensuing discussion, relevant inner products $\langle \cdot \vert \cdot \rangle$ and norms $\Vert  . \Vert$ induced by them will be inferred from the context. 

 An oracle implementing a function $f$  is described by the unitary operator $\textbf{O}_f$ which acts on a basis state  $\vert g\rangle \otimes\vert y\rangle \in \mathscr{H}$ by :
\begin{equation*}
\textbf{O}_f:\vert g\rangle \otimes\vert y\rangle \mapsto\vert g\rangle \otimes\vert y+f(g)\rangle 
\end{equation*}
Hence oracle evaluation of $f$ on the state of the query register $\mathbb{C}[G]$ is added to the response register $\mathbb{C}[X]$.

\begin{defn}
A \textit{query} is a unit norm state in $\mathscr{H}$, presented to the  oracle for evaluation. The set of queries is then: $\{\vert \Psi \rangle  \in \mathscr{H} :   \Vert \vert \Psi \rangle \Vert = 1\}$.
\end{defn}

Queries of interest to us are assumed to be in an equal superposition state over the group but arbitrary in the response register. We refer to them as the equal superposition tensor product (ESTP) queries. Denote this class of queries  $Q_0 \subset \mathscr{H}$.
\begin{equation*}
	 Q_0 := \{\frac{1}{\sqrt{| G |}}\bigg(\sum_{g\in G}\vert g\rangle \bigg) \otimes\vert v \rangle \ : \vert v \rangle\ \in \mathbb{C}[X] , \Vert \vert v \rangle\Vert = 1\}
\end{equation*}

We restrict the algorithms to those comprising  the following steps:
\begin{enumerate}[label=(\roman{*})]
\item \label{ists} Prepare a query  $\vert \Psi \rangle \in Q_0$.
\item Evaluate the oracle on $\vert \Psi \rangle$.
\item Measure the query register using a measurement (POVM) $\mathscr{E} = \{E_k\}_{k \in \mathscr{K}}$. Observe the outcome $k$ and decide upon the  corresponding  $H_k \in \mathscr{S}$ as the hidden subgroup.
\end{enumerate}

A query $\vert \Psi \rangle = \frac{1}{\sqrt{| G |}}\big(\sum_{g\in G}\vert g\rangle \big) \otimes\vert v \rangle \in Q_0$ is determined by the tensor factor $\vert v \rangle \in \mathbb{C}[X]$. This prompts the definition of  a \textit{slate}.

\begin{defn}
Let $\vert \Psi \rangle=\frac{1}{\sqrt{| G |}}\big(\sum_{g\in G}\vert g\rangle\big) \otimes\vert v \rangle$ be an ESTP query.  Its response register part $\vert v \rangle$  is the \textit{slate} of the query  $\vert \Psi \rangle$. 
\end{defn} 

We define the set of slates $S_0 := \{\vert v \rangle \in \mathbb{C}[X] : \Vert\vert v \rangle \Vert_{X} =1\}$. By definition, $S_0$ can be identified with $Q_0$:
\begin{equation} \label{qsident}
\iota_{_{S_0}} : S_0 \leftrightarrow Q_0, \quad\vert v \rangle \mapsto  \frac{1}{\sqrt{| G |}}\bigg(\sum_{g\in G}\vert g\rangle \bigg) \otimes\vert v \rangle
\end{equation}

Given this identification, we will refer to queries $\vert \Psi \rangle= \iota_{_{S_0}}\vert v \rangle = \frac{1}{\sqrt{| G |}}(\sum_{g\in G}\vert g \rangle)\otimes\vert v \rangle$ by their  slate  $\vert v \rangle$ and vice versa.

We recall the Fourier transform on $X$:
 \begin{equation} \label{fourx}
\mathcal{F}_X: \vert y \rangle \mapsto \frac{1}{\sqrt{D}}\sum_{x \in X} \omega^{y.x}\vert x \rangle =:\vert \omega^{-y} \rangle
\end{equation}
where $ y=(y_j), x=(x_j) \in X$, $\omega =(\omega_j)$,  $\omega_j=e^{i2\pi/{p^{\alpha_j}_j}}$ and $\omega^{y.x} = \omega^{ y_1 x_1}_1 \ldots \omega^{y_m x_m}_k$.

A slate $\vert v \rangle$ can be written in terms of the characters,  $\vert \omega^{-y} \rangle $, of $X$.
\begin{equation} \label{initstrw}
\vert v \rangle =\sum_{y \in X}  \beta_{v,y}\vert \omega^{-y}\rangle  \quad \text{where}  \quad \beta_{v,y} := \langle  \omega^{-y} \vert v \rangle
\end{equation}

%%%%%%%%%%%%%%%%%%%%%%%%%%%%%%%%%%%%%%%%%%%%%%%%%%%%%%%%%%%%%%%%%%%%%%%%%%%%%%%%%%%%%%%%%%%%%%%%%%%%%%%%%%%%

\section{Subgroup States} \label{section:clq}
States  arising from  functions $f  \in \mathscr{F}_k$ (constant and distinct on the left cosets of $H_k$) are  in some sense described by  the same subgroup $H_k$ and a reasonable measurement should target that subgroup.  Before we make this precise, let us consider the state  after the oracle implementing some  function  $f$ (not necessarily in $ \mathscr{F}_k$) has acted on a query $\vert \Psi \rangle = \frac{1}{\sqrt{| G |}}\bigg(\sum_{g\in G}\vert g\rangle \bigg) \otimes\vert v \rangle \in Q_0$. %We denote this state  $\vert v_f  \rangle$.
\begin{equation*}
\textbf{O}_f\vert \Psi \rangle = \sum_{y \in X} \beta_{v,y}\frac{1}{\sqrt{| G |}}\bigg(\sum_{g\in G}\omega^{-y \cdot f(g)}\vert g\rangle \bigg) \otimes\vert \omega^{-y} \rangle 
\end{equation*}
where $\{\beta_{v,y}\}_{y \in X}$ are as in \eqref{initstrw}. We denote the mixed state of the query register by $\rho_{_{f,v}}$:
\begin{align} \label{qregst}
\rho_{_{f,v}} &:=  \text{tr}_{_{\mathbb{C}[X]}}\left( \textbf{O}^{\phantom\dag}_f\vert \Psi \rangle \langle \Psi \vert \textbf{O}^\dag_f \right) \nonumber \\
&=\frac{1}{|G|}\sum_{y \in X} |\beta_{v,y}|^2 \bigg( \sum_{g,g' \in G}\omega^{-y \cdot (f(g)-f(g'))}\vert  g  \rangle\langle g'  \vert \bigg)
\end{align}

Recall the definition of a measurement as relevant to our discussion.  A measurement (POVM) on the query register $\mathbb{C}[G]$ is described by a set of  operators  $\mathscr{E} := \{E_k \}_{k \in \mathscr{K}}$ where $E_k \in \text{End}(\mathbb{C}[G])$ and satisfy:
%\begin{equation} 
\begin{enumerate}[label=(\roman{*})]
\item $E_k \geq 0   \:\: \forall  k \in \mathscr{K}$
\item
$\sum_{k \in \mathscr{K}} E_k = \textbf{I}$
\end{enumerate}
%\end{equation}
Here $\textbf{I}$ is the identity operator on $\mathbb{C}[G]$. 

If the state of the query register is given by a density operator $\rho_{_G} \in \text{End}(\mathbb{C}[G])$, then the  outcome $k \in \mathscr{K}$ is observed with probability  $\text{tr} ( E_k \rho_{_G})$. By choice, the measurement operator $E_k$ is associated with the subgroup $H_k$, so that a measurement outcome $k$ corresponds to the subgroup  $H_k$.
 
Given a slate $\vert v \rangle$ and an oracle implementing the function  $f$, the mixed state of the query register  $\mathbb{C}[G]$ after oracle evaluation is  $\rho_{_{f,v}}$ \eqref{qregst}.  The probability of observing the outcome $k'$ is $\text{tr} ( E_{k'} \rho_{_{f,v}})$. Now assume that the oracle hides the subgroup $H_k$. Since all the $f  \in \mathscr{F}_k$  are assumed equally likely,  the probability of outcome  $k'$  given that the oracle hides the subgroup $H_k$,  is described by the following probability function:
\begin{equation*} %\label{probcondid}
\mu_{v,\mathscr{E}}(k'|k):= \frac{1}{|\mathscr{F}_k|} \sum_{f\in \mathscr{F}_k}  \:  \text{tr} \big( E_{k'} \rho_{_{f,v}}\big) = \text{tr} \bigg(E_{k' }\frac{1}{|\mathscr{F}_k|} \sum_{f\in \mathscr{F}_k}  \:  \rho_{_{f,v}}\bigg)
\end{equation*}
This leads us to the following notion:
\begin{defn}
A \textit{subgroup state} for the subgroup $H_k$ and slate $\vert v \rangle$ is a mixed state,  obtained by averaging $\{\rho_{_{f,v}}\}_{f \in \mathscr{F}_k} $  from \eqref{qregst}. Denote this state $\rho_{_{k,v}}$.
\begin{align} \label{sbgpst}
\rho_{_{k,v}} &:=  \frac{1}{|\mathscr{F}_k|} \sum_{f\in \mathscr{F}_k}  \:  \rho_{_{f,v}} 
\end{align}
\end{defn} 

$\mu_{v,\mathscr{E}}(k'|k)$ can be written as a function of the subgroup state as:
\begin{equation}  \label{probcorid}
\mu_{v,\mathscr{E}}(k'|k)= \text{tr} \big(E_{k' } \:  \rho_{_{k,v}}\big)
\end{equation}

A subgroup state aggregates the mixed states resulting from oracle functions hiding a particular subgroup into a single mixed state. A measurement aims to distinguish such subgroup states.

To be able to work with oracle functions, however, we must make a few identifications. By definition, an oracle function $f  \in \mathscr{F}_k$, hiding the subgroup $H_k \in \mathscr{S}$,  factors through quotient by $H_k$.  Denote the quotient map by $q_{_k}$ (fixed by the choice of $k \in \mathscr{K}$).
\begin{align*} 
q_{_k}:  G \rightarrow  G/H_k, \quad g \mapsto gH_k
\end{align*}

To enumerate the various sets consistently, we can  identify $X$ with the set $\{0,\ldots,D-1\} \subset \mathbb{N}$ as follows:
\begin{equation*}
\iota_X : X \leftrightarrow \{0,\ldots,D-1\}, \quad x=(x_j) \mapsto  \sum^{m}_{j=1}  x_{_{j}} \prod^{j-1}_{i=1}p^{\alpha_i}_i
\end{equation*}

 For $n \in \{0,\ldots,D-1\} $, define the set of ``first $n$ elements in $X$":
\begin{equation*}
X_n \:\: := \iota^{-1}_X(\{0,\ldots,n-1\})
\end{equation*}
 For each $k \in \mathscr{K}$, we can  identify the cosets $G/H_k$ with  $X_{N_k}$. Fix such an identification $\iota_{_k}$.
\begin{align} \label{identcoset}
\iota_{_k}: X_{N_k} \leftrightarrow  G/H_k
\end{align}

With these constructions in hand, $f  \in \mathscr{F}_k$  can be written as a composition:
\begin{equation} \label{ffactor}
 f = \tilde{f}_k \circ q_{_k} = \gamma \circ \iota^{-1}_k \circ q_{_k}
 \end{equation}
 as shown in figure ~\ref{gammadiag}, where $\tilde{f} : G/H_k \rightarrow X$ and $\gamma : X_{N_k} \rightarrow X$  are the unique maps such that the diagram commutes.

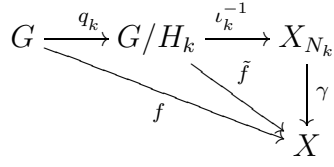
\begin{figure} [h]
\begin{displaymath} 
\xymatrix{ G \ar[r]^{q_{_k}} \ar[rrd]_{f} & G/H_k \ar[r]^{\iota^{-1}_k}  \ar[rd]^{\tilde{f}} &X_{N_k}  \ar[d]^{\gamma} \\
&&X}
\end{displaymath}
\caption{Oracle function $f$ factors}
\label{gammadiag}
\end{figure}

 It is then immediate that $\gamma$  is injective.  As $X_{N_k} \subseteq X$, $\gamma$  is also the restriction of some permutation of $X$  to  $X_{N_k}$. Let the set of such restrictions be given by the set $\Gamma_k$.
\begin{equation} \label{restrict}
\Gamma_k :=  \{ \sigma|_{X_{N_k}} : \sigma \in \text{S}_X\}
\end{equation}
 Denote by $\text{S}_X$ the group of permutations of $X$. Then for each $\gamma \in \Gamma_k$, the set:
 \begin{equation} \label{sgamma}
S_\gamma := \{ \sigma  \in \text{S}_X : \sigma|_{X_{N_k}} = \gamma\}
 \end{equation}
has cardinality $(D-N_k)!$.  Since such sets partition $\text{S}_X$, the number of possible oracle functions $f \in \mathscr{F}_k$ for every subgroup $H_k \in \mathscr{S}$ is:
\begin{equation*}
|\mathscr{F}_k| =  |\Gamma_k| = \frac{D!}{(D-N_k)!} 
\end{equation*}

With our specific factorization, an oracle function $f$ can be given as follows:
\begin{align} \label{ident}
f   \in \cup_{_{k \in \mathscr{K}}} \mathscr{F}_k &\leftrightarrow (k,\gamma) \in \mathscr{K}  \times \Gamma_k 
\end{align}
where $\gamma \in \Gamma_k $  is the unique map such that the diagram in figure ~\ref{gammadiag} commutes.

Our first result pertains to the form of the subgroup states $\{ \rho_{_{k,v}}\}_{k \in \mathscr{K}}$.

\begin{thm} \label{mainthm}
Let the subgroup hidden by the oracle be $H_k \in \mathscr{S}$. Let  $\vert v \rangle$ be a slate. Then the subgroup state $\rho_{_{k,v}}$ is given by an element  $\varphi_{_{k,v}} \in \mathbb{C}[G]$ acting by the right regular representation : 
\begin{align*}  
\rho_{_{k,v}} &= R(\varphi_{_{k,v}}) \nonumber \\
\varphi_{_{k,v}} &= |\beta_{v,0}|^{2} \vert \varphi_{_0} \rangle+ (1- |\beta_{v,0}|^{2})  \vert \varphi_{_{k,\bar0}} \rangle
\end{align*}
where $\beta_{v,0}= \langle  \omega^{0} \vert v \rangle$ is as defined in \eqref{initstrw}, and:
\begin{align*} 
\vert \varphi_{_0} \rangle &= \frac{1}{|G|}  \sum_{g \in G }\vert g \rangle \nonumber \\
 \vert \varphi_{_{k,\bar0}} \rangle &= \frac{1}{|G|} \bigg(  \frac{D}{(D-1)}\sum_{h \in H_k }\vert h \rangle -\frac{1}{(D-1)}\sum_{g \in G }\vert g \rangle  \bigg)   
\end{align*}

\end{thm}

\begin{proof}
We begin by defining  the coset state for $gH_k \in G/H_k$:
\begin{equation*} 
\vert gH_k \rangle :=\frac{1}{\sqrt{| H_k |}}\sum_{h\in H_k}\vert gh\rangle
\end{equation*}
For an oracle function $f  \in \mathscr{F}_k$, we rewrite the mixed  state of the query register \eqref{qregst} in terms of the coset states, taking account of the oracle function  property that  it  factors through  quotient by $H_k$. Hence, using  \eqref{ffactor} we can write:
\begin{equation*} 
\rho_{_{f,v}} =\frac{1}{N_k}\sum_{y \in X} |\beta_{v,y}|^2  \sum_{c,c' \in G/H_k}\omega^{-y \cdot (\tilde{f}(c)-\tilde{f}(c'))}\vert  c  \rangle\langle c'  \vert 
\end{equation*}
Under the identifications \eqref{identcoset} and \eqref{ident}:  $f \leftrightarrow (k,\gamma)$. We can recast $\rho_{_{f,v}}$ as:
\begin{equation*}
\rho_{_{f,v}} = \frac{1}{N_k}\sum_{y \in X} |\beta_{v,y}|^2  \sum_{r,r' \in X_{N_k}}\omega^{-y \cdot (\gamma(r)-\gamma(r'))}\vert  \iota_{_k}(r)  \rangle\langle \iota_{_k}(r')  \vert 
\end{equation*}
Averaging over all $\gamma \in  \Gamma_k$ from \eqref{restrict}, results in the subgroup state $\rho_{_{k,v}}$ \eqref{sbgpst}.
\begin{align*} \label{rholosiggg}
\rho_{_{k,v}} &= \frac{(D-N_k)!}{D!} \sum_{\gamma \in \Gamma_k} \rho_{_{f,v}} \nonumber \\
		&= \sum_{y \in X}  |\beta_{v,y}|^2 \frac{(D-N_k)!}{D!}  \sum_{\gamma \in \Gamma_k} \frac{1}{{N_k}} \sum_{r,r' \in X_{N_k}}\omega^{-y \cdot (\gamma(r)-\gamma(r'))}\vert  \iota_{_k}(r)  \rangle\langle \iota_{_k}(r') \vert
\end{align*}
Since the sets $S_\gamma$ \eqref{sgamma} all have the cardinality $(D-N_k)!$,  we can average over $\text{S}_X$ instead of $\Gamma_k$. Then:
\begin{align*} 
\rho_{_{k,v}} &= \sum_{y \in X}  |\beta_{v,y}|^2\bigg(\frac{1}{D!} \sum_{\sigma \in \text{S}_X} \frac{1}{{N_k}} \sum_{r,r' \in X_{N_k}}\omega^{-y \cdot (\sigma(r)-\sigma(r'))}\vert  \iota_{_k}(r)  \rangle\langle \iota_{_k}(r') \vert \bigg)
\end{align*}

For $y \in X$, we define:
\begin{equation*}  
\check\rho_{_{k,y}} =\frac{1}{D!} \sum_{\sigma \in \text{S}_X} \frac{1}{{N_k}} \sum_{r,r' \in X_{N_k}}\omega^{-y \cdot (\sigma(r)-\sigma(r'))}\vert  \iota_{_k}(r)  \rangle\langle \iota_{_k}(r') \vert
\end{equation*}
Then we write $\rho_{_{k,v}}$ as:
\begin{equation}  \label{rholosig}
\rho_{_{k,v}} =\sum_{y \in X}  |\beta_{v,y}|^2\check\rho_{_{k,y}}
\end{equation}
We consider the operators $\check\rho_{_{k,y}}$ above for $y \in X$.  When $y = 0$:
\begin{equation} \label{varphily}
\check\rho_{_{k,0}} =\bigg(\frac{1}{\sqrt{| G |}}\sum_{g\in G}\vert g\rangle\bigg)  \bigg(\frac{1}{\sqrt{| G |}}\sum_{g\in G}\langle g  \vert \bigg) 
\end{equation}
which is simply the  projection on the  equal superposition state in $ \mathbb{C}[G]$.

When $y \neq 0$, we find that:
\begin{align*} 
\check\rho_{_{k,y}} &=\frac{1}{D!} \sum_{\sigma \in \text{S}_X} \frac{1}{N_k} \bigg( \sum_{r \in X_{N_k}}\vert  \iota_{_k}(r)  \rangle \langle \iota_{_k}(r)  \vert  +  \sum_{r \neq r' \in X_{N_k}} \omega^{-y \cdot (\sigma(r)-\sigma(r'))}\vert  \iota_{_k}(r)  \rangle \langle \iota_{_k}(r')  \vert \bigg) \nonumber \\
&= \frac{1}{N_k} \Bigg( \sum_{r \in X_{N_k}}\vert  \iota_{_k}(r)  \rangle \langle \iota_{_k}(r)  \vert  +  \sum_{r \neq r' \in X_{N_k}}\vert  \iota_{_k}(r)  \rangle \langle \iota_{_k}(r')  \vert  \bigg(\frac{1}{D!}   \sum_{\sigma \in \text{S}_X} \omega^{-y \cdot (\sigma(r)-\sigma(r'))}\bigg)  \Bigg)  
\end{align*}
The set $\{ r \neq r' \in X_{N_k}\}$ can be written as $X_{N_k}\times X_{N_k}\setminus \Delta(X_{N_k})$, where $\Delta$ denotes the diagonal map defined for any set $S$ as:
\begin{equation*} 
\Delta: S \rightarrow S \times S, \quad s \mapsto (s,s)
\end{equation*}
We compute the sum  in the inside bracket in the expression for $\check\rho_{_{k,y}}$ above. For a pair $(r,r') \in X_{N_k}\times X_{N_k}\setminus \Delta(X_{N_k})$, define the set  $\Omega_{(r,r')}$ of pairs obtained by evaluating  permutations in $\text{S}_X$  on the pair $(r,r')$.
\begin{equation*}
\Omega_{(r,r')} := \{(\sigma(r),\sigma(r')) \:\: : \:\: \sigma \in \text{S}_X\} \subseteq  X \times X \setminus \Delta(X)
\end{equation*}
We make two observations. Firstly, $\Omega_{(r,r')} = X \times X \setminus \Delta(X)$.  Secondly, for every pair $(s,s') \in X \times X \setminus \Delta(X)$ the set $\{ \sigma  \in \text{S}_X : (\sigma(r),\sigma(r')) = (s,s')\}$ has cardinality $(D-2)!$, and such sets partition $\text{S}_X$. They imply:
\begin{align*} 
\frac{1}{D!}   \sum_{\sigma \in \text{S}_X} \omega^{-y \cdot (\sigma(r)-\sigma(r'))} &= \frac{(D-2)!}{D!}   \sum_{(s,s') \in \Omega_{(r,r')}} \omega^{-y \cdot (s-s')}   \nonumber \\
 &=\frac{(D-2)!}{D!}   \sum_{(s,s') \in X \times X \setminus \Delta(X)} \omega^{-y \cdot (s-s')}  \nonumber \\
&=  \frac{(D-2)!}{D!} \sum_{x \in  X\setminus\{0\}} D \omega^{-y \cdot x} \nonumber \\
 &=  -\frac{(D-2)!}{D!} D \nonumber \\
&=   -\frac{1}{(D-1)} \nonumber\\
\end{align*}
where the second to last  equality is straightforward from the fact that any nontrivial irreducible character of a finite group sums  to $0$.
We continue with $\check\rho_{_{k,y}}$ simplification:
\begin{align*}  
\check\rho_{_{k,y}} &= \frac{1}{N_k} \bigg( \sum_{r \in X_{N_k}}\vert  \iota_{_k}(r)  \rangle \langle \iota_{_k}(r)  \vert    -\frac{1}{(D-1)}\sum_{r \neq r' \in X_{N_k}}\vert  \iota_{_k}(r)  \rangle \langle \iota_{_k}(r')  \vert    \bigg) \\
 &= \frac{1}{N_k} \bigg( \frac{D}{(D-1)}\sum_{r \in X_{N_k}}\vert  \iota_{_k}(r)  \rangle \langle \iota_{_k}(r)  \vert    -\frac{1}{(D-1)}\sum_{r , r' \in X_{N_k}}\vert  \iota_{_k}(r)  \rangle \langle \iota_{_k}(r')  \vert    \bigg)
\end{align*}
With the identification in \eqref{identcoset}, this becomes:
\begin{align*}  
\check\rho_{_{k,y}}  &= \frac{1}{N_k} \bigg( \frac{D}{(D-1)}\sum_{c \in G/H_k}\vert  c  \rangle \langle c  \vert    -\frac{1}{(D-1)}\sum_{c , c' \in G/H_k}\vert  c  \rangle \langle c' \vert    \bigg)
\end{align*}
It is worth noticing in the above expression that:
\begin{equation*}  
\sum_{c \in G/H_k}\vert  c  \rangle \langle c \vert 
\end{equation*}
is the projection on the span of coset states associated with $H_k$, and 
\begin{equation*}  
\sum_{c , c' \in G/H_k}\vert  c  \rangle \langle c' \vert  = N_k \bigg( \frac{1}{\sqrt{| G |}}\sum_{g\in G}\vert g\rangle\bigg)  \bigg(\frac{1}{\sqrt{| G |}}\sum_{g\in G}\langle g  \vert \bigg) = N_k \: \check\rho_{_{k,0}}
\end{equation*}.
where $\check\rho_{_{k,0}}$ is as in \eqref{varphily}.

We see that  $\check\rho_{_{k,y}}$ is independent of $y$, and $\check\rho_{_{k,0}}$ is also independent of $k$. Consequently, we define:
\begin{align}  \label{varphibardef}
\check\rho_{_0} &:=\bigg( \frac{1}{\sqrt{| G |}}\sum_{g\in G}\vert g\rangle\bigg)  \bigg(\frac{1}{\sqrt{| G |}}\sum_{g\in G}\langle g  \vert \bigg)  \nonumber \\
\check\rho_{_{k,\bar0}} &:= \frac{1}{N_k} \bigg( \frac{D}{(D-1)}\sum_{c \in G/H_k}\vert  c  \rangle \langle c  \vert    -\frac{1}{(D-1)}\sum_{c , c' \in G/H_k}\vert  c  \rangle \langle c' \vert    \bigg) 
\end{align}
from which follows that we can write $\rho_{_{k,v}}$ in  \eqref{rholosig} as:
\begin{align*} 
\rho_{_{k,v}} &=\sum_{y \in X}  |\beta_{v,y}|^2\check\rho_{_{k,y}}  \nonumber \\
&= |\beta_{v,0}|^{2} \check\rho_{_0} +  (1- |\beta_{v,0}|^{2}) \check\rho_{_{k,\bar0}}
\end{align*}

One readily checks that the operators $\check\rho_{_0} $ in  \eqref{varphily} and $\check\rho_{_{k,\bar0}} $  in \eqref{varphibardef} commute with the left translation \eqref{rtact}:
\begin{align*} 
L^{-1}(z) \check\rho_{_0} L({z})  &= \check\rho_{_0} \\
L^{-1}(z) \check\rho_{_{k,\bar0}} L({z}) & = \check\rho_{_{k,\bar0}}   \quad \forall  z \in G
\end{align*} 
which implies that  $\check\rho_{_0} $ and $\check\rho_{_{k,\bar0}} $ are given by elements of  $\mathbb{C}[G]$ acting by the right regular representation, found by evaluating the two operators at the identity of $\mathbb{C}[G]$. We write:
\begin{align*}  
\check\rho_{_0} &= R(\vert \varphi_{_0} \rangle) \nonumber \\
\check\rho_{_{k,\bar0}} &= R( \vert \varphi_{_{k,\bar0}} \rangle)  
\end{align*}
where: 
\begin{align*} 
\vert \varphi_{_0} \rangle &= \frac{1}{|G|}  \sum_{g \in G }\vert g \rangle \nonumber \\
 \vert \varphi_{_{k,\bar0}} \rangle &= \frac{1}{|G|} \bigg(  \frac{D}{(D-1)}\sum_{h \in H_k }\vert h \rangle -\frac{1}{(D-1)}\sum_{g \in G }\vert g \rangle  \bigg)   
\end{align*}
\end{proof}

%%%%%%%%%%%%%%%%%%%%%%%%%%%%%%%%%%%%%%%%%%%%%%%%%%%%%%%%%%%%%%%%%%%%%%%%%%%%%%%%%%%%%%%%%%%%%%%%%%%%%%%%%%%%

\section{The Character Query and Optimal Measurements} \label{section:cq}

To compare ESTP queries, we would like to quantify a query by  the highest success probability it can achieve with any measurement. In that connection, we denote by $\mathscr{E}_G$ the set of all POVMs on the query register:  $\mathscr{E}_G := \{ \mathscr{E} :  \mathscr{E} \:\: \text{a POVM on } \mathbb{C}[G]\}$.

%%%%%%%%%%%%%%%% NOT NECESSARY ? %%%%%%%%%%%%%%%

\begin{comment}

Two or more measurements may have the same measurement probabilities in \eqref{probcorid}.  We call such measurements  equivalent.
\begin{defn}
Measurements  $\mathscr{E} = \{E_k\}_{k \in \mathscr{K}},\mathscr{U} = \{U_k\}_{k \in \mathscr{K}}\in \mathscr{E}_G$ are  \textit{equivalent}  if:
\begin{equation*}
\mu_{v,\mathscr{E}}(k'|k)=\mu_{v,\mathscr{U}}(k'|k) \quad \forall   \vert v \rangle \in S_0,  \: k,k' \in \mathscr{K}
\end{equation*}
where $\mu_{v,\mathscr{E}}(k'|k)$ is as in \eqref{probcorid}.
\end{defn} 

\end{comment}
%%%%%%%%%%%%%%%%%%%%%%%%%%%%%%%%%%%%%%%%%

Fix a  slate $\vert v \rangle$.  Given that the oracle hides the subgroup  $H_k$, the probability of correctly identifying it using a measurement $\mathscr{E} = \{E_k\}_{k \in \mathscr{K}}$ is $\mu_{v,\mathscr{E}}(k|k) =   \text{tr}  \big(E_{k } \:  \rho_{_{k,v}}\big)$ by \eqref{probcorid}. All the subgroups in the set of subgroups $\mathscr{S}$ are equally likely to be hidden by the oracle, with a uniform prior probability over $\mathscr{S}$. So  the probability of successful subgroup identification  denoted $ S_{v,\mathscr{E}}$ is:
\begin{equation*}
 S_{v,\mathscr{E}} = \frac{1}{K} \sum_{k \in \mathscr{K}}  \mu_{v,\mathscr{E}}(k|k) = \frac{1}{K} \sum_{k \in \mathscr{K}} \: \text{tr}  \big(E_{k } \:  \rho_{_{k,v}}\big)
 \end{equation*}
This is exactly the probability of successful state discrimination in \eqref{sucs}   in which the  measurement $\mathscr{E} = \{E_k\}_{k \in \mathscr{K}}$ is used to  distinguish the subgroup states $\{\rho_{_{k,v}}\}_{k \in \mathscr{K}}$ distributed with prior probability $\{p_k\}_{k \in \mathscr{K}}$ given by $p_k = 1/K$ (section \ref{subsection:hspsd}). As $\mathscr{E}$ varies over  $\mathscr{E}_G$, we get  a function $S_{v}$ on the set of POVMs $\mathscr{E}_G$, giving the probability of successful subgroup identification.
\begin{defn}
The \textit{success probability} of a slate $\vert v \rangle$, denoted by  $S_{v}$, is a function on  $\mathscr{E}_G$:
\begin{equation*} %\label{Pc}
\begin{array}{rcl}
S_{v}: \mathscr{E}_G&\longrightarrow& [0,1] \nonumber \\
\mathscr{E} = \{E_k\}_{k \in \mathscr{K}} &\longmapsto  &S_{v,\mathscr{E}} = \frac{1}{K} \sum_{k \in \mathscr{K}} \: \text{tr}  \big(E_{k } \:  \rho_{_{k,v}}\big)
\end{array}
\end{equation*}	
\end{defn} 
We call $S_{v}(\mathscr{E})$ the success probability of the slate $\vert v \rangle$ for the measurement $\mathscr{E}$.

To see the  topological structure of  $\mathscr{E}_G$, a useful alternate description of a  measurement  is:
\begin{equation*}
\mathscr{E} := (E_k)_{k \in \mathscr{K}} \in {\text{End}(\mathbb{C}[G])}^{\times K}
\end{equation*}
where the $\{E_k \}_{k \in \mathscr{K}}$ satisfy the conditions in section \ref{section:clq}.  $\text{End}(\mathbb{C}[G])$ is  a finite dimensional Hilbert space. Give ${\text{End}(\mathbb{C}[G])}^{\times K}$  the product topology. Then $\mathscr{E}_{_G}$ is a compact subset of ${\text{End}(\mathbb{C}[G])}^{\times K}$. Given a slate $\vert v \rangle$, $S_{v}$ is continuous. Hence we can define the maximum probability, over $\mathscr{E}_{_G}$, of correctly determining a hidden subgroup with  a slate  $\vert v \rangle$.
\begin{defn}
The \textit{optimum success probability} of  a slate  $\vert v \rangle$, denoted by $\hat{S}_{v}$, is:
\begin{align*}
	 \hat{S}_{v}:=\max_{\mathscr{E} \in \mathscr{E}_{_G}}\{S_{v}(\mathscr{E})\}
\end{align*}
\end{defn} 

By the definition of  $ \hat{S}_{v}$, there  is some measurement that achieves it.  Such measurements may not be unique. 
\begin{defn}
A measurement $\mathscr{E} \in \mathscr{E}_G$ is an \textit{optimal measurement}  for a slate  $\vert v \rangle$  if  $S_{v}(\mathscr{E}) = \hat{S}_{v}$.
\end{defn} 

We will also need a basic result from the representation theory of finite groups~\cite{wall:sri}. This is summarized in Appendix \ref{apprt}. Let $\hat G$ be the equivalence classes of irreducible unitary representations of $G$. Fix a representation $(\pi^{\lambda}, V^{\lambda})$ in the class $\lambda$ for each $\lambda \in \hat G$. Its dual representation is denoted by $(\pi^{\lambda^*}, V^{\lambda^*})$. Let the dimension of $V^{\lambda}$ be $d_{_\lambda}$.  Then the right translation \eqref{rtact} is isomorphic to:
\begin{equation} \label{rtrrep}
R(z) \cong  \bigoplus_{\lambda \in \hat{G}} \textbf{I}_{V^{{\lambda}^{*}}}\otimes \pi^{\lambda}(z)
\end{equation}
where $\textbf{I}_{V^{\lambda^*}}$ is the  identity operator on the space  ${V^{\lambda^*}}$.
We denote the trivial representation  by  $(\pi^0, V^0)$.

 Next, we define some specific queries that will be needed in the computation of  success probabilities.

\begin{defn} The character query, the standard query, and the  constant query, are defined by their slates and the identification $\iota_{_{S_0}}$ in \eqref{qsident}.
\begin{equation*}
\begin{array} {lrclrcl}
\text{The }  constant \:\: query:&\vert \Psi_0  \rangle &:=& \iota_{_{S_0}}\vert v_0 \rangle,&  \vert v_0 \rangle&:=&\vert \omega^{(0,0,\ldots,0)} \rangle \\
\\
\text{The } standard \:\: query:&\vert \Psi_s \rangle &:=& \iota_{_{S_0}}\vert v_s \rangle,& \vert v_s \rangle&:=& \vert (0,0,\ldots,0) \rangle \\
\\
\text{The } character  \:\: query:&\vert \Psi_c\rangle &:=& \iota_{_{S_0}}\vert v_c \rangle,&\vert v_c \rangle&:=& \vert \omega^{(-1,0,\ldots,0)} \rangle 
\end{array}
\end{equation*}
where  $\vert  \omega^{-y} \rangle$ is as in \eqref{fourx}.
\end{defn}

We express the success probability of any given ESTP query as a function of  that of the character query, and show that the optimal measurements are common to almost all the queries.  

\begin{cor}   \label{optquery}  
 \begin{enumerate}[label=(\roman{*})] 
\item \label{measrr} For any  measurement  $\mathscr{U} = \{U_k\}_{k \in \mathscr{K}}$,  there is  a measurement $\mathscr{E} = \{E_k\}_{k \in \mathscr{K}}$   of the form:
\begin{equation*} 
E_{k} \cong  \bigoplus_{\lambda \in \hat{G}} \textbf{I}_{V^{{\lambda}^{*}}}\otimes E^{\lambda}_{k}, \quad \quad E^{\lambda}_{k} \in \begin{rm}{End}\end{rm}(V^{\lambda})
\end{equation*}
hence given by elements of $\mathbb{C}[G]$ acting by the right regular representation, such that $\mathscr{U}$ and $\mathscr{E}$ have the same conditional probabilities in~\eqref{probcorid}. That is:
\begin{equation*}
\mu_{v,\mathscr{E}}(k'|k)=\mu_{v,\mathscr{U}}(k'|k) \quad \forall   \vert v \rangle \in S_0,  \: k,k' \in \mathscr{K}
\end{equation*}
%where $\mu_{v,\mathscr{E}}(k'|k)$ is as in~\eqref{probcorid}.
In particular, $S_{v}(\mathscr{U}) = S_{v}(\mathscr{E})  \:\: \forall   \vert v \rangle \in S_0$. 

\item  \label{optmeas} Let $\vert v \rangle$ be a slate. Given a measurement $\mathscr{E}  \in \mathscr{E}_G$, the success probability of  $\vert v \rangle$ for $\mathscr{E}$, $S_{v}(\mathscr{E})$, is:
\begin{align*} 
	S_{v}(\mathscr{E})&= \frac{|\beta_{v,0}|^{2} }{K} +   (1-|\beta_{v,0}|^{2})  S_{v_c}(\mathscr{E}) 
 \end{align*}
In particular, if a  measurement  is optimal  for some slate  $\vert v \rangle$ such that $|\beta_{v,0}| \neq 1$ $(\vert v \rangle \notin \{e^{i\theta}\vert v_0 \rangle : \theta \in \mathbb{R}\})$, then it is optimal for every  slate.
\end{enumerate}
\end{cor}
\begin{proof}

\ref {measrr}:  Suppose we are given a slate $\vert v \rangle$, and a  measurement $\mathscr{U} = \{U_k\}_{k \in \mathscr{K}}$. Then  $\mu_{v,\mathscr{U}}(k'|k) = \text{tr} ( U_{k'} \rho_{_{k,v}})$.   

Each operator $\rho_{_{k,v}}$ is isomorphic to a direct sum  by Theorem \ref{mainthm}   and   \eqref{rtrrep}. Define $\tilde{\rho}_{_{k,v}}$ as:
\begin{equation} \label{rholrtrep}
\rho_{_{k,v}} \cong  \bigoplus_{\lambda \in \hat{G}} \textbf{I}_{V^{{\lambda}^{*}}}\otimes \pi^{\lambda}(\varphi_{_{k,v}}) =: \tilde{\rho}_{_{k,v}}
\end{equation}

We have that: the trace of a linear operator on a finite dimensional vector space is invariant under vector space isomorphisms,  $\tilde{\rho}_{_{k,v}}$ is a direct sum, and $\text{End}(V^{{\lambda}^{*}}\otimes V^{\lambda}) \cong \text{End}(V^{{\lambda}^{*}})\otimes \text{End}(V^{\lambda})$. Thus,  there exists  $\tilde{E}_{k'} \in \bigoplus_{\lambda \in \hat{G}} \text{End}(V^{{\lambda}^{*}})\otimes \text{End}(V^{\lambda})$ such that :
$\text{tr} ( \tilde{E}_{k'} \tilde{\rho}_{_{k,v}}) = \text{tr} ( U_{k'} \rho_{_{k,v}})$. Trace,  
$
\text{tr}: \text{End}(V^{{\lambda}^{*}})\otimes \text{End}(V^{\lambda}) \rightarrow \mathbb{C}
$,  
has the property:
\begin{equation*}
\text{tr}(A\otimes B) = \text{tr} (A) \text{tr} (B)  \quad \forall  A \in \text{End}(V^{{\lambda}^{*}}), \:B \in \text{End}(V^{\lambda})
\end{equation*}
This,  together with  the form of the summands $\textbf{I}_{V^{{\lambda}^{*}}}\otimes \pi^{\lambda}(\varphi_{_{k,v}})$ in  \eqref{rholrtrep},  implies that  $\tilde{E}_{k'}$ can be chosen so that:
\begin{equation*} 
\tilde{E}_{k'} =  \bigoplus_{\lambda \in \hat{G}} \textbf{I}_{V^{{\lambda}^{*}}}\otimes E^{\lambda}_{k'}
\end{equation*}
for some  $E^{\lambda}_{k'} \in \begin{rm}{End}\end{rm}(V^{\lambda})$. Under the isomorphism  \eqref{rtrrep}, we can find $E_{k'}$ :
\begin{equation*} 
E_{k'}  \cong \:\: \tilde{E}_{k'}= \bigoplus_{\lambda \in \hat{G}} \textbf{I}_{V^{{\lambda}^{*}}}\otimes E^{\lambda}_{k'} 
\end{equation*}

\ref{optmeas}: Let $\mathscr{E} = \{E_k\}_{k \in \mathscr{K}}$.
%We can assume that the measurement $\mathscr{E} = \{E_k\}_{k \in \mathscr{K}}$ is given  by:
%\begin{equation*} 
%E_{k} \cong  \bigoplus_{\lambda \in \hat{G}} \textbf{I}_{V^{{\lambda}^{*}}}\otimes E^{\lambda}_{k}
%\end{equation*}
%where $E^{\lambda}_{k} \in \begin{rm}{End}\end{rm}(V^{\lambda})$. Using Theorem \ref{mainthm} and \eqref{rholrtrep},
Using Theorem \ref{mainthm}, we can write $S_{v}(\mathscr{E})$ as:

\begin{align} \label {Ppsi}
	S_{v}(\mathscr{E}) &=  \frac{1}{K} \sum_{k \in \mathscr{K}} \: \text{tr}  \big(E_{k } \:  \rho_{_{k,v}}\big) \nonumber \\
 &= |\beta_{v,0}|^{2} \frac{1}{K}   \sum_{k \in \mathscr{K}}  \text{tr}  \big(  E_{k}   \check\rho_{_0} \big) +  (1-|\beta_{v,0}|^{2})\frac{1}{K}   \sum_{k \in \mathscr{K}}  \text{tr}  \big(  E_{k}   \check\rho_{_{k,\bar0}} \big) 
\end{align}

Notice that the definition of  $\vert v_c \rangle$ makes $\beta_{0,v_c} = 0$. By Theorem \ref{mainthm},   $\rho_{_{k,v_c}} = \check\rho_{_{k,\bar0}} = R( \vert \varphi_{_{k,\bar0}} \rangle)$,  which makes  $S_{v_c}(\mathscr{E})$:
%\begin{align} \label{svc}
%S_{v_c}(\mathscr{E}) &=   \frac{1}{K} \sum_{\lambda \in \hat{G}} d_{_\lambda}  \sum_{k \in \mathscr{K}}  \text{tr}  \bigg(  E^{\lambda}_{k}    \pi^{\lambda}\big( \vert \varphi_{_{k,\bar0}} \rangle\big) \bigg)
%	 \end{align}
\begin{align} \label{svc}
S_{v_c}(\mathscr{E}) &=   \frac{1}{K}   \sum_{k \in \mathscr{K}}  \text{tr}  \big(  E_{k}   \check\rho_{_{k,\bar0}} \big)
	 \end{align}
Because $\sum_{k \in \mathscr{K}} E_{k}  = \textbf{I}$ and $\text{tr}  \big( \check\rho_{_0} \big) = 1$, 
\begin{equation*} 
\frac{1}{K}   \sum_{k \in \mathscr{K}}  \text{tr}  \big(  E_{k}   \check\rho_{_0} \big) = \frac{1}{K}    \text{tr}  \big( \check\rho_{_0} \big) = \frac{1}{K}
\end{equation*}

%We make use of  Theorem \ref{mainthm}.    $\check\rho_{_0}=R(\vert \varphi_{_0} \rangle)$ is the projection on the trivial character. Hence:
%\begin{align*} 
%\pi^0(\vert \varphi_{_0} \rangle) &= \frac{1}{|G|}  \sum_{g \in G } \pi^0(g) =1 
%\end{align*}
%For the non-trivial representations $\lambda \in \hat G \setminus \{0\}$:
%\begin{align*} 
%\pi^{\lambda}(\vert \varphi_{_0} \rangle) &=0  
%\end{align*}
This  simplifies  $S_{v}(\mathscr{E})$ in  \eqref{Ppsi}.
\begin{align} \label {Ppsirw}
	S_{v}(\mathscr{E})
	%&= \frac{1}{K}  \left( |\beta_{v,0}|^{2} \sum_{k \in \mathscr{K}}  \: \text{tr} ( E^0_{_{k}}  ) +    (1-|\beta_{v,0}|^{2}) S_{v_c}(\mathscr{E})\right)  \nonumber \\
	&= \frac{|\beta_{v,0}|^{2} }{K} +     (1-|\beta_{v,0}|^{2}) S_{v_c}(\mathscr{E})
\end{align}

\end{proof}

%%%%%%%%%%%%%%%%%%%%%%%%%%%%%%%%%%%%%%%%%%%%%%%%%%%%%%%%%%%%%%%%%%%%%%%%%%%%%%%%%%%%%%%%%%%%%%%%%%%%%%%%%%%%

\section{Success Probability of the Character Query} \label{section:oq}
We are ready to  show that the character query has the maximum success probability,  strictly higher than that of the standard query. Once we have shown this, we take a closer look at the Deutsch's problem  for single bit.

\begin{cor} \label{charquery} 
The optimum success probabilities satisfy:
\begin{align*} 
	 \frac{1 }{K} \leq  \hat{S}_{v}  \leq   \hat{S}_{v_c}  \quad \forall\vert v \rangle \in S_0
 \end{align*}
The lower equality is true if and only if $|\beta_{v,0}| = 1$ $(\vert v \rangle \in  \{e^{i\theta}\vert v_0 \rangle : \theta \in \mathbb{R}\})$, and the upper equality is true if and only if $\beta_{v,0} = 0$. In particular,  $\hat{S}_{v_s} <  \hat{S}_{v_c} $.
\end{cor}

\begin{proof}
First, define the trivial measurement  to be $\mathscr{T}=\{T_k\}$,  $T_k= \frac{1}{K} \textbf{I}$, where $\textbf{I}$ is the identity operator on $\mathbb{C}[G]$. The success probability of the trivial measurement, for any slate $\vert v \rangle$, is $S_{v}(\mathscr{T}) =  \frac{1}{K}$. This is because $\text{tr}(\rho_{_{k,v}})=1  \:\: \forall k \in \mathscr{K}$. 

By Corollary \ref{optquery} \ref{optmeas}, we just need to  show that there exists a measurement $\mathscr{M} \in \mathscr{E}_G$ which has a success probability of  the character query for $\mathscr{M}$, $S_{v_c}(\mathscr{M})$, greater than $\frac{1}{K}$.   Since the trivial measurement $\mathscr{T}$ has success probability  $\frac{1}{K}$, it is sufficient to show:
\begin{equation} \label{chkcond}
\exists \:\: \mathscr{M} \in \mathscr{E}_G \:\: :  \quad S_{v_c}(\mathscr{T}) < S_{v_c}(\mathscr{M}) 
\end{equation}

Before defining $\mathscr{M}$, we  consider the subgroup states for the character query, $\{\rho_{_{k,v_c}}\}_{k \in \mathscr{K}}$. As seen in Corollary \ref{optquery}, $\rho_{_{k,v_c}} = R( \vert \varphi_{_{k,\bar0}} \rangle)$. Decompose   each $ \vert \varphi_{_{k,\bar0}} \rangle$ into its Fourier components,  denoted by $\tilde{\varphi}^{\lambda}_{k}$ for  $\lambda \in \hat G$.
\begin{align}  \label{varphibartriv}
 \tilde{\varphi}^{0}_{k} &:= \pi^0( \vert \varphi_{_{k,\bar0}} \rangle) \nonumber \\
 &= \frac{1}{|G|} \bigg(  \frac{D}{(D-1)}\sum_{h \in H_k } \pi^0(h) -\frac{1}{(D-1)}\sum_{g \in G } \pi^0(g)    \bigg)   \nonumber \\
&= \frac{1}{|G|} \bigg(  \frac{D}{(D-1)}|H_k| -\frac{1}{(D-1)}|G|   \bigg)   \nonumber \\
&=   \frac{D-N_k}{(D-1)N_k}
\end{align}
For every non-trivial representation $\lambda \in \hat G \setminus \{0\}$:
\begin{align}  \label{varphibarnontriv}
\tilde{\varphi}^{\lambda}_{k} &:=  \pi^{\lambda}( \vert \varphi_{_{k,\bar0}} \rangle) \nonumber \\
&= \frac{1}{|G|}   \frac{D}{(D-1)}\sum_{h \in H_k } \pi^{\lambda}(h)    \nonumber \\
&=  \frac{D}{(D-1)N_k} \frac{1}{|H_k|} \sum_{h \in H_k } \pi^{\lambda}(h)     	
\end{align}
Also define:  
\begin{equation*}
\tilde{\varphi}^{\lambda} :=
 \frac{1}{K} \sum_{k \in \mathscr{K}} \tilde{\varphi}^{\lambda}_{k}
 \end{equation*} 

Note that: 
\begin{align*}
 \text{Ker}(\tilde{\varphi}^{\lambda}) &= \cap_{k \in \mathscr{K}} \text{Ker}( \tilde{\varphi}^{\lambda}_{k} ) \\
\text{Im}(\tilde{\varphi}^{\lambda}) &= V^{\lambda} \setminus \text{Ker}(\tilde{\varphi}^{\lambda})
 \end{align*}

We  take the measurement  $\mathscr{M} =  \{M_k\}$  to be the Pretty Good Measurement from section \ref{subsection:hspsd} \eqref{pgmdef}, for the subgroup states $\{\rho_{_{k,v_c}}\}_{k \in \mathscr{K}}$ with prior probability $\{p_k\}_{k \in \mathscr{K}}$ given by $p_k = 1/K$. Using  \eqref{rholrtrep}  and the definitions above we can write the subgroup states as:
\begin{equation*}
\rho_{_{k,v_c}} \cong  \bigoplus_{\lambda \in \hat{G}} \textbf{I}_{V^{{\lambda}^{*}}}\otimes \tilde{\varphi}^{\lambda}_{k}
\end{equation*}
and compute the PGM:
\begin{equation*} 
M_k \cong  \bigoplus_{\lambda \in \hat{G}} \textbf{I}_{V^{{\lambda}^{*}}}\otimes M^{\lambda}_k
\end{equation*}
where: 
\begin{equation*}
M^{\lambda}_k = \frac{1}{K} \left({(\tilde{\varphi}^{\lambda})}^{-1/2} \tilde{\varphi}^{\lambda}_{k} {(\tilde{\varphi}^{\lambda})}^{-1/2}\oplus \textbf{I}_{\text{Ker}(\tilde{\varphi}^{\lambda})} \right)
\end{equation*}
Here $\textbf{I}_{\text{Ker}(\tilde{\varphi}^{\lambda})}$ is the identity operator on $\text{Ker}(\tilde{\varphi}^{\lambda})$, and  ${(\tilde{\varphi}^{\lambda})}^{-1/2}$ is:
\begin{equation*}
{(\tilde{\varphi}^{\lambda})}^{-1/2} := \left({\tilde{\varphi}^{\lambda}\mid_{\text{Im}(\tilde{\varphi}^{\lambda})}}\right)^{-1/2}\oplus \textbf{I}_{\text{Ker}(\tilde{\varphi}^{\lambda})}
\end{equation*}

From \eqref{svc}, the success probabilities of the character query for the PGM $\mathscr{M}$ and the trivial measurement $\mathscr{T}$ are:
\begin{align*} 
S_{v_c}(\mathscr{M}) &= \frac{1}{K} \sum_{\lambda \in \hat{G}}  d_{_\lambda} \sum_{k \in \mathscr{K}}  \text{tr}  \left(  M^{\lambda}_k    \tilde{\varphi}^{\lambda}_{k} \right)  = \frac{1}{K} \sum_{\lambda \in \hat{G} }  d_{_\lambda} \sum_{k \in \mathscr{K}}   \frac{1}{K}\text{tr}  \left( {(\tilde{\varphi}^{\lambda})}^{-1/2} \tilde{\varphi}^{\lambda}_{k} {(\tilde{\varphi}^{\lambda})}^{-1/2}   \tilde{\varphi}^{\lambda}_{k} \right ) \nonumber \\
S_{v_c}(\mathscr{T}) &=  \frac{1}{K} \sum_{\lambda \in \hat{G}}  d_{_\lambda}\sum_{k \in \mathscr{K}}  \text{tr}  \left(  T^{\lambda}_{k}    \tilde{\varphi}^{\lambda}_{k}\right) = \frac{1}{K} \sum_{\lambda \in \hat{G} }  d_{_\lambda}   \text{tr}  \left(  \tilde{\varphi}^{\lambda}  \right)
\end{align*}

We recall the inner product and norm on $M_{d_{_\lambda}}(\mathbb{C})$:
\begin{align*} 
<A,B> &=  \text{tr}  \left(  A  B^\dag \right)  \\
\lVert A \rVert &=  \sqrt{\text{tr}  \left(  A  A^\dag \right)}  
\end{align*}
where $ A,B \in M_{d_{_\lambda}}(\mathbb{C})$. 
This allows  $S_{v_c}(\mathscr{M})$ and $S_{v_c}(\mathscr{T})$ to be expressed as:
\begin{align*} 
S_{v_c}(\mathscr{M}) &=  \frac{1}{K} \sum_{\lambda \in \hat{G} }  d_{_\lambda} \sum_{k \in \mathscr{K}}   \frac{1}{K} \lVert{(\tilde{\varphi}^{\lambda})}^{-1/4} \tilde{\varphi}^{\lambda}_{k} {(\tilde{\varphi}^{\lambda})}^{-1/4} {\rVert}^2 \nonumber \\
S_{v_c}(\mathscr{T})  &= \frac{1}{K} \sum_{\lambda \in \hat{G} }  d_{_\lambda}  \lVert {(\tilde{\varphi}^{\lambda})}^{1/2}   {\rVert}^2
\end{align*}
Then the  inequality in \eqref{chkcond} follows from these additional observations:
\begin{enumerate}[label=(\roman{*})]
\item $\exists \lambda \in \hat{G}$, and $k,k' \in \mathscr{K}$ such that $\tilde{\varphi}^{\lambda}_{k} \neq \tilde{\varphi}^{\lambda}_{k'}$
\item  
$
\sum_{k \in \mathscr{K}}  {(\tilde{\varphi}^{\lambda})}^{-1/4} \tilde{\varphi}^{\lambda}_{k} {(\tilde{\varphi}^{\lambda})}^{-1/4} = K \: {(\tilde{\varphi}^{\lambda})}^{1/2} \quad \:\: \forall  \lambda \in \hat{G}
$
\end{enumerate}
and the next lemma.
\end{proof}

\begin{lemma} Let $V$ be a finite dimensional Hilbert space over $\mathbb{C}$. Let $\{v_i\}^n_{i=1}$ be a set of vectors in $V$. Then:
\begin{align*}
   \frac{1}{n} \lVert \sum^n_{i=1} v_i {\rVert}^2 \leq  \sum^n_{i=1}  \lVert v_i {\rVert}^2 
\end{align*}
with 
equality if and only if $v_i = v_j \:\: \forall  i,j \in \{1,\ldots,n\}$.
\end{lemma}
\begin{proof}
\begin{align*} 
n \sum^n_{i=1}  \lVert v_i {\rVert}^2 -   \lVert \sum^n_{i=1} v_i {\rVert}^2 &= n \sum^n_{i=1}  \lVert v_i {\rVert}^2 -   \sum^n_{i,j=1} \langle v_i,v_j \rangle \\
&= (n-1) \sum^n_{i=1}  \lVert v_i {\rVert}^2 -   \sum_{1\leq i<j \leq n}\left( \langle v_i,v_j \rangle + \langle v_j,v_i \rangle \right)\\
&\geq (n-1) \sum^n_{i=1}  \lVert v_i {\rVert}^2 -   \sum_{1\leq i<j \leq n}2 \lVert v_i \rVert  \lVert v_j \rVert \\
&=   \sum_{1\leq i<j \leq n}  \lVert v_i {\rVert}^2 -  2 \lVert v_i \rVert  \lVert v_j \rVert +  \lVert v_j {\rVert}^2 \\
&=   \sum_{1\leq i<j \leq n}  \left(\lVert v_i {\rVert} -    \lVert v_j {\rVert}\right)^2 
\end{align*}
where the inequality  above is a result of  the Cauchy-Schwarz inequality. The lemma follows.
\end{proof}

%%%%%%%%%%%%%%%%%%%% NOT NECESSARY  PGM  SUB-RESULT  ?%%%%%%%%%%%
\begin{comment}
\begin{remark}
The proof of Corollary \ref{charquery} actually shows that $\mathscr{M} =  \{M_{k}\}^L_{l=1}$, can be defined for any set of positive operators $\{K_k\}^L_{l=1}$ on a finite dimensional Hilbert space $V$. Let 
\begin{equation*}
K := \frac{1}{L} \sum^L_{l =1} K_k
 \end{equation*} 
 and assume $K$ has rank $\text{dim}(V)$. Then the PGM is given by:
\begin{equation*}
M_k := \frac{1}{L} K^{-1/2} K_k {K}^{-1/2}
\end{equation*}
and satisfies:
\begin{equation*}
\sum^L_{l =1} \: \text{tr} ( M_k K_k ) \geq \text{tr} (K)
\end{equation*}
with equality if and only if $K_i = K_j \:\: \forall i,j$.
\end{remark}
\end{comment}
%%%%%%%%%%%%%%%%%%%%%%%%%%%%%%%%%%%%%%%%%%%%%%%%%%%%

\begin{remark}
In view of Corollary  \ref{optquery}, an optimal measurement  $\mathscr{\hat{E}}=\{\hat{E}_k\}$ can be described by:
\begin{equation*}  
\hat{E}_{k} \cong  \bigoplus_{\lambda \in \hat{G}} \textbf{I}_{V^{{\lambda}^{*}}}\otimes \hat{E}^{\lambda}_{_{k}}
\end{equation*}
Using  \eqref{varphibartriv} and \eqref{varphibarnontriv}:
\begin{align} \label{optpc}
\hat{S}_{v_c} &=  \frac{1 }{K} \sum_{k \in \mathscr{K}}\bigg[ \frac{D-N_k}{(D-1)N_k}  + \frac{D}{(D-1) |G| } \sum_{\lambda \in \hat{G}\setminus\{0\}} d_{_\lambda}  \sum_{h \in H_k }  \text{tr}  \big(\hat{E}^{\lambda}_{k}    \pi^{\lambda}(h) \big) \bigg]
\end{align}
Therefore, the character query performs better as the dimension ($D$) of the response register decreases.
\end{remark}

\begin{remark}
From the definition of  the standard query, $\beta_{_{0,v_s}} = 1/\sqrt{D}$. We conclude from  Corollary \ref{optquery} \ref{optmeas} that  $\lim_{D->\infty} \hat{S}_{v_c}=  \hat{S}_{v_s}$. So the optimum success probability of the character query decreases to that of the standard query as the dimension of the response register increases. 
\end{remark}

We now turn our attention to the Deutsch's problem for single bit ($n=1$) and understand it in our framework.  Let us  restate  the set up from  section \ref{subsection:hspsd}. The group $G = \mathbb{Z}/(2)$.  The response space is $X = \mathbb{Z}/(2)$. Possible hidden subgroups are: $H_0 := G $ and $H_1 := \{0\}$. Constant functions hide the subgroup $H_0$ and balanced functions hide $H_1$. The subgroup indices are $N_0=1$ and $N_1 = 2$ respectively. The system on which the oracle and measurement act is:
 \begin{equation*}
 \mathscr{H} :=  \mathbb{C}[G]\otimes\mathbb{C}[X] = \mathbb{C}[\mathbb{Z}/(2)]\otimes\mathbb{C}[\mathbb{Z}/(2)]
 \end{equation*}
where the first tensor factor $\mathbb{C}[G]$ is the query register, and the second factor $\mathbb{C}[X]$ is the response register. Since $X = \mathbb{Z}/(2)$, this makes the  dimension of the response register  $D=2$.  The group $G = \mathbb{Z}/(2)$ has two $1$-dimensional representations: the trivial representation and the alternating representation.
\begin{equation*}  
 \left .
 \begin{array}  {ll}
\pi^0(x) &= 1 \\
\pi^{-}(x) &= (-1)^x
\end{array}
\right\}  \quad \forall x \in \mathbb{Z}/(2)
\end{equation*}  
Phase kickback harnesses the character query: $\vert \Psi_c\rangle = \vert + \rangle\vert - \rangle = \frac{1}{\sqrt{2}} (\vert 0 \rangle +\vert 1 \rangle)\otimes\vert \omega^{-1} \rangle$,
where $\omega = -1$, and $\vert v_c\rangle =\vert \omega^{-1} \rangle =\vert - \rangle = \frac{1}{\sqrt{2}} (\vert 0 \rangle - \vert 1 \rangle)$.
After the oracle evaluation, we can use  \eqref{varphibartriv} and \eqref{varphibarnontriv} to deduce the Fourier components of the subgroup states $\rho_{_{k,v_c}} = R( \vert \varphi_{_{k,\bar0}}\rangle)$.
\begin{comment}
\begin{equation*}
\begin{array}{lll}
\pi^0( \vert \varphi_{_{0,\bar0}} \rangle) &=  1 \\
\\
\pi^{-}( \vert \varphi_{_{0,\bar0}} \rangle) &= 0  \\
\\
\pi^0(\vert \varphi_{_{1,\bar0}} \rangle) &= 0 \\
\\
\pi^{-}(\vert \varphi_{_{1,\bar0}} \rangle) &= 1  \\
\end{array}
\end{equation*}
\end{comment}

\begin{align*}
\rho_{_{0,v_c}} &\cong \pi^0( \vert \varphi_{_{0,\bar0}} \rangle) \oplus \pi^{-}( \vert \varphi_{_{0,\bar0}} \rangle) =  1\oplus 0\\
\rho_{_{1,v_c}} &\cong \pi^0(\vert \varphi_{_{1,\bar0}} \rangle) \oplus \pi^{-}(\vert \varphi_{_{1,\bar0}} \rangle) =  0\oplus 1
\end{align*}\
We choose a measurement $\mathscr{E} := \{E_k \}_{k \in \{0,1\}}$   as follows:
 \begin{equation*}  
E_{k} \cong  {E}^{0}_{_{k}} \oplus {E}^{-}_{_{k}} 
\end{equation*}
where:
\begin{align*}
E_0 \cong {E}^{0}_{_{0}}\oplus{E}^{-}_{_{0}} =  1\oplus 0 \\
E_1 \cong {E}^{0}_{_{1}}\oplus{E}^{-}_{_{1}} =  0 \oplus 1
\end{align*}
By  \eqref{svc}, this measurement has a  probability of success $ S_{v_c}(\mathscr{E}) = 1$.  This is precisely what we call the measurement for Deutsch's problem in section \ref{subsection:pk}.  By using the Hadamard transform $\textbf{H}$, the  character basis of $\mathbb{C}[\mathbb{Z}/(2)]$ is rotated to the computational basis, i.e.:
\begin{equation*}
\textbf{H}: \left \{
\begin{array}{lll}
 \frac{1}{\sqrt{2}}(\vert 0 \rangle + \vert 1 \rangle) &\mapsto&\vert 0 \rangle \\
\\
 \frac{1}{\sqrt{2}}(\vert 0 \rangle - \vert 1 \rangle) &\mapsto&\vert 1 \rangle  \\
\end{array} \right .
\end{equation*}
Measuring in the computational basis is then equivalent to the measurement above. It is because the subgroup $H_1 = \{0\}$ has an index the same as  the dimension of the response register, $N_1 = D$, that it has no projection on the  trivial representation. In contrast, the subgroup $H_0 = G$ has an index less than  the dimension of the response register, $N_0 <  D$, which results in its producing a non-zero projection on the trivial representation. That is why the subgroups can be distinguished  with probability $1$. This illustrates the part played by the dimension of the response register relative to the subgroup indices, and explains the algorithm  in  the group representation context. 

%%%%%%%%%%%%%%%%%%%%%%%%%%%%%%%%%%%%%%%%%%%%%%%%%%%%%%%%%%%%%%%%%%%%%%%%%%%%%%%%%%%%%%%%%%%%%%%%%%%%%%%%%%%%

\section{Conjugate Subgroups} \label{section:consub}
We specialize further to a  class of subgroups consisting of conjugates of a particular subgroup $H \leq G$, i.e., $\mathscr{S}=\{g^{-1}Hg : g \in G\}$, and determine the optimum success probability for any $\vert v \rangle  \in S_0$. To that end, we exploit a result of Moore and Russell~\cite{cmar:pgm}. $Plancherel \; measure$, denoted by ${\mu}_{_P}$, is the probability distribution on $\hat G$ defined as:
\begin{equation*}
{\mu}_{_P}(\lambda) := d^2_\lambda/|G| \quad \text{for} \: \lambda \in \hat G
\end{equation*}
In particular:
\begin{equation*}
{\mu}_{_P}(0) = 1/|G| 
\end{equation*}

As in Moore and Russell~\cite{cmar:pgm}, define the set  $\Lambda_H \subseteq \hat G$:
\begin{equation*}
\Lambda_H := \big\{\lambda \in \hat G \: : \: \frac{1}{|H|} \sum_{h \in H } \pi^{\lambda}(h) \neq 0 \big\}
\end{equation*}
We recognize $\frac{1}{|H|} \sum_{h \in H } \pi^{\lambda}(h) $ as the projection onto  the  space of $H$-invariants in  $\lambda \in \hat G$.  

Denote the normalizer of $H$ in $G$ by ${\begin{rm}N\end{rm}}_G(H)$. 

\begin{cor} \label{conjsub}
Let  $H \leq G$, and $\mathscr{S}=\{gHg^{-1} : g \in G\}$.    Let $N:=[G:H]$ and $N_{\begin{rm}C\end{rm}}:=[G:{\begin{rm}N\end{rm}}_G(H)]$.  Then for any slate $\vert v \rangle$, the optimum success probability, $ \hat{S}_{v}$,  is given by:
  \begin{align*} 
	 \hat{S}_{v}&= \bigg(|\beta_{v,0}|^{2} -  (1-|\beta_{v,0}|^{2})  \frac{1}{(D-1)}\bigg) \frac{1}{N_{\begin{rm}C\end{rm}}}+   (1-|\beta_{v,0}|^{2})   \frac{D}{(D-1)}\frac{|H|}{N_{\begin{rm}C\end{rm}}}{\mu}_{_P}(\Lambda_H) 
\end{align*}
\end{cor}

\begin{proof} Let the subgroup indexing set be $\mathscr{K}=G/{\begin{rm}N\end{rm}}_G(H)$. Then $K = |\mathscr{K}| =  N_{\begin{rm}C\end{rm}}$. As $\beta_{_{0,v_s}} = 1/\sqrt{D}$ (the case of standard query), by Corollary \ref{optquery} \ref{optmeas}  an optimal measurement for $\vert v_s \rangle$ is optimal for any slate $\vert v \rangle$. Moore and Russell~\cite{cmar:pgm} have shown that  the Pretty Good Measurement (PGM)  for the standard query is such a measurement, and also derived its success probability.  Using their result:
\begin{align*}
	  \hat{S}_{v_s}  =  \frac{|H|}{N_{\begin{rm}C\end{rm}}}{\mu}_{_P}(\Lambda_H)
 \end{align*}
 From Corollary  \ref{optquery} \ref{optmeas}:
  \begin{align*} 
	 \hat{S}_{v_s}&= \frac{|\beta_{_{0,v_s}}|^2 }{N_{\begin{rm}C\end{rm}}} +   (1-|\beta_{_{0,v_s}}|^2)    \hat{S}_{v_c} \nonumber \\
			&= \frac{1}{DN_{\begin{rm}C\end{rm}}} +    \frac{(D-1)}{D}    \hat{S}_{v_c} 
\end{align*}
Together they imply:
  \begin{align*}
	 \hat{S}_{v_c}&= \frac{D}{(D-1)}\frac{|H|}{N_{\begin{rm}C\end{rm}}}{\mu}_{_P}(\Lambda_H) - \frac{1}{(D-1)N_{\begin{rm}C\end{rm}}}
 \end{align*}
Applying Corollary \ref{optquery} \ref{optmeas} again, we get the result we seek.
\end{proof}

%\begin{remark}
%By Corollary \ref{charquery}, the character query  success probability:
%  \begin{align*}
%	 \hat{S}_{v_c}&= \frac{D}{(D-1)}\frac{|H|}{N_{\begin{rm}C\end{rm}}}{\mu}_{_P}(\Lambda_H) - \frac{1}{(D-1)N_{\begin{rm}C\end{rm}}}
% \end{align*}
% is greater than that of the standard query  (Moore and Russell~\cite{cmar:pgm}):
% \begin{align*}
%	  \hat{S}_{v_s}  =  \frac{|H|}{N_{\begin{rm}C\end{rm}}}{\mu}_{_P}(\Lambda_H)
 %\end{align*}
%\end{remark}

\begin{remark}
Unlike in Deutsch's problem, here the subgroup states for the character query all have the same projection on the trivial representation, by  \eqref{varphibartriv}.  These contribute nothing toward distinguishing the subgroups,  and  are  eliminated if the dimension of the response register is the same as the  index of $H$, i.e., if $D=N$.  Assuming such is the case,  from the proof of Corollary \ref{conjsub}, the character query  succeeds with probability:
  \begin{align*} 
	 \hat{S}_{v_c}=   \frac{N}{(N-1)}\frac{|H|}{N_{\begin{rm}C\end{rm}}}{\mu}_{_P}(\Lambda_H\setminus {\{0\}}) 
\end{align*}
\end{remark}
%%%%%%%%%%%%%%%%%%%%%%%%%%%%%%%%%%%%%%%%%%%%%%%%%%%%%%%%%%%%%%%%%%%%%%%%%%%%%

\section{Conclusion and Discussion} \label{section:conclusion}
We have addressed the problem of query selection for the single-query hidden subgroup problem (HSP) over a general finite group $G$ and an abelian response space $X$. Our results indicate that for a single-query algorithm with measurements of the query register, and among the class of queries in an equal superposition state  over the group, we can maximize the subgroup identification (success) probability using a query that has no projection on the  constant query (defined as the equal superposition over the group as well as the response space). The character query is  an example of this set of queries contained in  the  $|X| -1$ dimensional subspace (in the above discussion $D=|X|$) orthogonal to the constant query.

This generalization of  the phase kickback trick explains the phase kickback for Deutsch's problem (single bit) in representation theoretic terms.  It arises naturally  when we analyze how the success probability of the algorithm depends on the choice of the query. Imposing some structure (an abelian structure in this discussion) on $X$ is necessary to analyze the effect of different queries. The result that the optimal measurements for  algorithms  in our class are common to all ESTP queries other than the phase multiples of the constant query, is not something one would expect a priori.  It shows why in Deutsch's problem the phase kickback trick  and the standard method have the same optimal measurement. The character query outperforms  the query used in the standard method single-query HSP algorithms, and gives an improvement over the success probability of Moore and Russell~\cite{cmar:pgm} for conjugate subgroups.

For  the character query itself,  the success probability  decreases as  the response register dimension $|X|$ increases; an example  is the conjugate subgroups case (where the highest success probability is achieved when the subgroups have the same index as the  response register dimension). The  response register dimension relative  to the subgroup indices differentiates the subgroups through their projection on  the trivial representation. This has the potential to improve the success probability, as we saw in the analysis of phase kickback in Deutsch's problem, which is somewhat more complicated than the conjugate subgroups case.

Our approach towards optimizing single queries for HSP depends on  conceptualizing the  oracle functions as given by permutations. By  computing with the response space $X$,  and developing and interpreting results with respect to the group representation one gains insight about queries and oracle action. Finally, recognizing that the structure of the problem allows the use of   representation theory in conjunction with   the PGM (in general a sub-optimal measurement),  and the analysis of  measurements using norms, leads to the  proof of  optimality of the character query. 

We expect aspects of this approach, in particular the resulting generalized phase kickback, to have applications in other domains, among them multi-query settings as in  Bacon, Childs and Van Dam~\cite{bcv:opd} and Meyer and Pommersheim~\cite{mp:mqqs}.

%%%%%%%%%%%%%%%%%%%%%%%%%%%%%%%%%%%%%%%%%%%%%%%%%%%%%%%%%%%%%%%%%%%%%%%%%%%%%%%%%%%%%%%%%%%%%%%%%%%%%%
\section*{Acknowledgments}
I gratefully acknowledge  enlightening  discussions with, and  important suggestions by   Prof.\ David Meyer and Prof.\ Nolan Wallach. I would also like to thank Dr.\ Orest Bucicovschi and Prof.\ James Pommersheim for the very helpful exchanges of ideas. This work has been partially supported by the Defense Advanced Research Projects Agency as part of the Quantum Entanglement Science
and Technology program under grant N66001-09-1-2025.

 %%%%%%%%%%%%%%%%%%%%%%%%%%%%%%%%%%%%%%%%%%%%%%%%%%%%%%%%%%%%%%%%%%%%%%%%%%%%%%%%%%%%%%%%%%%%%%%%%%%%%%

\newpage
%\input{bib_query.tex}

%% END MATTER
% \printindex %% Uncomment to display the index
% \nocite{}  %% Put any references that you want to include in the bib 
%               but haven't cited in the braces.
% \bibliographystyle{alpha}  %% This is just my personal favorite style. 
%                              There are many others.
% \bibliography{myrefs}  %% This looks for the bibliography in myrefs.bib 

\begin{bibdiv}
%\begin{biblist}
\begin{biblist}[\normalsize]

\def\MR#1{\relax\ifhmode\unskip\spacefactor3000 \space\fi
  \href{http://www.ams.org/mathscinet-getitem?mr=#1}{MR#1}}

\bib{nc:qcqi}{book}{
      author={Nielsen, M.},
      author={Chuang, I.},
      title={Quantum Computation and Quantum Information},
   publisher={Cambridge University Press},
        date={2000},
        ISBN={0-52163503-9},
}


\bib{sp:ptfd}{article}{
      author={Shor, P.},
      title={Polynomial-Time Algorithms for Prime Factorization and Discrete Logarithms on a Quantum Computer},
      journal={SIAM Journal on Computing},
       volume={26},
       issue={5},
       date={1997},
       pages={1484-1509}
}
 
 
 
 \bib{eh:qogi}{article}{
      author={Ettinger, M.},
      author={H\o yer, P.},
      title={A Quantum Observable for the Graph Isomorphism Problem},
      date={1999},
 note={\href{http://arxiv.org/abs/quant-ph/9901029}{arXiv:quant-ph/9901029}},
}


 
\bib{ro:qclp}{article}{
      author={Regev, O.},
      title={Quantum Computation and Lattice Problems},
       journal={SIAM Journal on Computing},
       volume={33},
       issue={3},
       date={2004},
       pages={738-760}
}


\bib{ah:sdqs}{article}{
      author={Holevo, A.},
       title={Statistical Decision Theory for Quantum Systems},
       journal={Journal of Multivariate Analysis},
       volume={3},
       date={1973},
pages={337-394},
}


\bib{ykl:opt}{article}{
      author={Yuen, H.},
      author={Kennedy, R.},
      author={Lax, M.},
       title={Optimum testing of multiple hypotheses in quantum detection theory},
       journal={IEEE Transactions on Information Theory},
       volume={IT-21(2)},
       date={1975},
pages={125-134},
}


\bib{ip:sao}{article}{
      author={Ip, L.},
       title={Shor's algorithm is optimal},
note={\href{http://lawrenceip.com/papers/hspsdpabstract.html}{lawrenceip.com/papers/hspsdpabstract.html}},
}


 \bib{ak:qmas}{article}{
      author={Kitaev, A.},
      title={Quantum measurements and the Abelian Stabilizer Problem},
       date={1996},
       journal={Electronic Colloquium on Computational Complexity (ECCC)},
       volume={3},
 %note={\href{http://arxiv.org/abs/quant-ph/1008.0010v1}{arXiv:quant-ph/9511026v1}},
 }
	





\bib{me:hspee}{article}{
      author={Mosca, M.},
      author={Eckert, A.},
       title={The Hidden Subgroup Problem and Eigenvalue Estimation on a Quantum Computer},
       journal={Quantum Computing and Quantum Communications},
      volume={1509},
	issue={May},
       date={1999},
       pages={174-188}
}
E
\bib{eh:qanc}{article}{
      author={Ettinger, M.},
      author={H\o yer, P.},
      title={On Quantum Algorithms for Noncommutative Hidden Subgroups},
      journal={Advances in Applied Mathematics},
       volume={25},
       issue={3},
       date={2000},
       pages={239-251},
}



\bib{kup:sed}{article}{
      author={Kuperberg, G.},
      title={A Subexponential-Time Quantum Algorithm for the Dihedral Hidden Subgroup Problem},
       journal={SIAM Journal on Computing},
      volume={35},
	issue={1},
       date={2005},
       pages={170-188}
% note={\href{http://arxiv.org/abs/quant-ph/0302112}{arXiv:quant-ph/0302112}},
}


\bib{reg:sep}{article}{
      author={Regev, O.},
       title={A Subexponential Time Algorithm for the Dihedral Hidden Subgroup Problem with Polynomial Space},
       date={2004},
 note={\href{http://arxiv.org/abs/quant-ph/0406151}{arXiv:quant-ph/0406151}},
}

\bib{bcv:opd}{article}{
      author={Bacon, D.},
      author={Childs, A.},
      author={Van Dam, V.},
      title={Optimal measurements for the dihedral hidden subgroup problem},
       journal={Chicago Journal of Theoretical Computer Science},
       date={2005},
}

\bib{bcv:omqa}{article}{
      author={Bacon, D.},
      author={Childs, A.},
      author={Van Dam, V.},
      title={From optimal measurements to efficient quantum algorithms for the  hidden subgroup problem over semidirect product groups},
       journal={46th Annual IEEE Symposium on Foundations of Computer Science},
       pages={469-478},
       date={2005},
} 

\bib{cmar:pgm}{article}{
      author={Moore, C.},
      author={Russell, A.},
       title={For Distinguishing Conjugate Hidden Subgroups, the Pretty Good Measurement is as Good as it Gets},
      date={2005},
 note={\href{http://arxiv.org/abs/quant-ph/0501177}{arXiv:quant-ph/0501177}},
}

\bib{bd:opscm}{article}{
      author={Bacon, D.},
      author={Decker, T.},
      title={The optimal single copy measurement for the hidden subgroup problem},
       journal={Physical Review A},
       volume={77},
       issue={3},
       date={2008},
}




\bib{kr:eqwh}{article}{
      author={Krovi, H.},
      author={R\"oteller, M.},
      title={An Efficient Quantum Algorithm for the Hidden Subgroup Problem over Weyl-Heisenberg Groups},
      journal = {Lecture Notes in Computer Science}
      date={2008},
      volume={5393},
      pages={70-88}
% note={\href{http://arxiv.org/abs/quant-ph/9901034}{arXiv:quant-ph/9901034}},
}



\bib{fw:hrp}{article}{
      author={Wang, F.},
       title={The Hidden Subgroup Problem},
      date={2010},
 note={\href{http://arxiv.org/abs/quant-ph/1008.0010v1}{arXiv:quant-ph/1008.0010v1}},
}

\bib{cl:hrp}{article}{
      author={Lomont, C.},
       title={The Hidden Subgroup Problem - Review and Open Problems},
      date={2004},
 note={\href{http://arxiv.org/abs/quant-ph/0411037}{arXiv:quant-ph/0411037}},
}


\bib{hrt:nsr}{article}{
      author={Hallgren, S.},
      author={Russell, A.},
      author={Ta-shma, A.},
      title={Normal Subgroup Reconstruction and Quantum
Computation Using Group Representations},
      journal={Proc. 32nd ACM Symp. on Theory of Computing},
%       volume={26},
%       issue={5},
pages = {627--635},
       date={2000},
}



\bib{db:cghg}{article}{
      author={Bacon, D.},
       title={How a Clebsch-Gordan Transform Helps to Solve the Heisenberg Hidden Subgroup Problem},
      date={2007},
 note={\href{http://arxiv.org/abs/quant-ph/0612107}{arXiv:quant-ph/0612107}},
}




\bib{mrrs:psfs}{article}{
      author={Moore, C.},
      author={Rockmore, D.},
     author={Russell, A.},
      author={Schulman, L.},
       title={The Power of Strong Fourier Sampling: Quantum Algorithms for Affine Groups and Hidden Shifts},
       journal={SIAM Journal on Computing},
       volume={37},
issue={3},
       date={2007},
pages={938-958},
}




\bib{hh:qft}{article}{
author = {L. Hales and S. Hallgren},
title = {An Improved Quantum Fourier Transform Algorithm and Applications},
journal ={Annual IEEE Symposium on Foundations of Computer Science},
year = {2000},
issn = {0272-5428},
pages = {515-525},
%doi = {http://doi.ieeecomputersociety.org/10.1109/SFCS.2000.892139},
%publisher = {IEEE Computer Society},
%address = {Los Alamitos, CA, USA},
}


\bib{mrr:psfs}{article}{
      author={Moore, C.},
      author={Rockmore, D.},
     author={Russell, A.},
       title={Generic quantum Fourier transforms},
       journal={ACM Transactions on Algorithms (TALG)},
       volume={2},
issue={4},
       date={2006},
pages={707-723},
}



\bib{dj:rspq}{article}{
      author={Deutsch, D.},
      author={Jozsa, R.},
       title={Rapid Solution of Problems by Quantum Computation},
       journal={Proceedings of the Royal Society of London A},
       volume={439},
       date={1992},
pages={553-558},
}




%\bib{reg:qcl}{article}{
 %     author={Regev, O.},
 %      title={Quantum computation and lattice problems},
 %      date={2002},
 %      journal = {Proceedings of the 43rd symposium on foundations of computer science},
  %     pages={520--530}
% note={\href{http://arxiv.org/abs/quant-ph/0406151}{arXiv:quant-ph/0406151}},
%}


%\bib{qnc:eh}{article}{
%author = {M. Ettinger and P. Hoyer},
%title = {On quantum algorithms for noncommutative hidden subgroups},
%journal ={Advances in Applied Mathematics},
%year = {2000},
%issn = {0196-8858},
%pages = {515},
%doi = {http://doi.ieeecomputersociety.org/10.1109/SFCS.2000.892139},
%publisher = {IEEE Computer Society},
%address = {Los Alamitos, CA, USA},
%}

\bib{rj:qaft}{article}{
      author={Jozsa, R.},
       title={Quantum Algorithms and the Fourier Transform},
       journal={Proceedings of the Royal Society of London A},
       volume={454},
       date={1998},
pages={323-337},
} 

\bib{cemm:qar}{article}{
      author={Cleve, R.},
      author={Ekert, E.},
      author={Macchiavello, C.},
      author={Mosca, M.},
       title={Quantum Algorithms Revisited},
       journal={Proceedings of the Royal Society of London A},
       volume={454},
       date={1998},
pages={339-354},
}



\bib{ykl:qar}{article}{
   author={Yuen, H. P.},
     author={Kennedy, R.S.},
     author={Lax, M.},
   title={Optimum testing of multiple hypotheses in quantum detection theory},
   journal={IEEE Trans. Inform. Theory},
   volume={21},
   date={1975},
   pages={125--134},
 %  issn={0002-9947},
   %review={\MR{0481100}},
}



\bib{wall:sri}{book}{
      author={Goodman, R.},
      author={Wallach, N. },
      title={Symmetry, Representations and Invariants},
      series={Graduate Texts in Mathematics},
   publisher={Springer},
        date={2009},
      volume={255},
        ISBN={978-0-387-79851-6},
}


\bib{mp:mqqs}{article}{
      author={Meyer, D.},
      author={Pommersheim, J.},
       title={Multi-query quantum sums},
    journal={preprint},
       date={2010},
}





\end{biblist}
\end{bibdiv}
%                          which should be formatted as a bibtex file.

%%%%%%%%%%%%%%%%%%%%%%%%%%%%%%%%%%%%%%%%%%%%%%%%%%%%%%%%%%%%%%%%%%%%%%%%%%%%%%%%%%%%%%%%%%%%%%%%%%%%%%
\newpage
\appendix
\section*{A result from the representation theory of finite groups} \label{apprt}
The group $G \times G$ acts on $\mathbb{C}[G]$ by left and right translations. Denote this representation by $\tau$:
\begin{equation*} 
\tau(z,w)\phi(g) = \phi(z^{-1} gw) \quad \text{for }  z,g,w \in G
\end{equation*}

Let $\hat G$ be the equivalence classes of irreducible unitary representations of $G$, and fix a representation $(\pi^{\lambda}, V^{\lambda})$ in the class $\lambda$ for each $\lambda \in \hat G$. Let the dimension of $V^{\lambda}$ be $d_{_\lambda}$.  The dual representation $(\pi^{\lambda^*}, V^{\lambda^*})$  is given by:
\begin{equation*} 
\langle\pi^{\lambda^*}(z)v^*,v \rangle = \langle v^*,\pi^{\lambda}(z^{-1})v \rangle
\end{equation*}
for $z \in G, v\in V^{\lambda}$, and $v^* \in V^{{\lambda}^{*}}$.

For $\lambda \in \hat G$, define $\vartheta_\lambda(v^*\otimes v)(g) = \langle v^*,\pi^{\lambda}(g)v \rangle$ for $g \in G, v^* \in V^{\lambda^*} \text{and} \: v \in V$. Extend $\vartheta_\lambda$ to a linear map from $V^{{\lambda}^{*}}\otimes V^{\lambda}$ to $\mathbb{C}[G]$. Under the action of $G \times G$, the space $\mathbb{C}[G]$ decomposes as:
\begin{equation*} 
\mathbb{C}[G] =  \bigoplus_{\lambda \in \hat{G}} \vartheta_\lambda(V^{{\lambda}^{*}}\otimes V^{\lambda})
\end{equation*}
With this decomposition, the left translation \eqref{rtact} is isomorphic to:
\begin{equation*}
L(z) \cong  \bigoplus_{\lambda \in \hat{G}}\pi^{{\lambda}^{*}}(z)\otimes  \textbf{I}_{V^{\lambda}}
\end{equation*}
and the right translation is isomorphic to:
\begin{equation*}
R(z) \cong  \bigoplus_{\lambda \in \hat{G}} \textbf{I}_{V^{{\lambda}^{*}}}\otimes \pi^{\lambda}(z)
\end{equation*}
where $\textbf{I}_{V^{\lambda}}$ and $\textbf{I}_{V^{\lambda^*}}$ are identity operators on the spaces $V^{\lambda}$ and ${V^{\lambda^*}}$ respectively.

\end{document}